\documentclass[english,british,american]{article}
\usepackage[T1]{fontenc}
\usepackage[latin9]{inputenc}
\usepackage{babel}
\usepackage{array}
\usepackage{float}
\usepackage{booktabs}
\usepackage{units}
\usepackage{multirow}
\usepackage{amstext}
\usepackage{graphicx}

\usepackage{breakurl}
\usepackage{hyperref}
\makeatletter

\providecommand{\tabularnewline}{\\}
\floatstyle{ruled}
\newfloat{algorithm}{tbp}{loa}
\providecommand{\algorithmname}{Algorithm}
\floatname{algorithm}{\protect\algorithmname}

\usepackage[noblocks]{authblk}

\usepackage[margin=1in]{geometry}

\usepackage{amssymb,amsfonts}

\usepackage{amsthm}
\theoremstyle{plain}
\newtheorem{thm}{Theorem}[section]
\newtheorem{lem}[thm]{Lemma}

\theoremstyle{definition}
\newtheorem{defn}{Definition}[section]

\theoremstyle{remark}

\usepackage[section]{placeins}
\makeatletter
\AtBeginDocument{%
  \expandafter\renewcommand\expandafter\subsection\expandafter{%
    \expandafter\@fb@secFB\subsection
  }%
}
\makeatother

\usepackage{ellipsis} 

\usepackage{nicefrac}
\usepackage{algorithmicx,algorithm,listings,algcompatible,algpseudocode}
\algrenewcommand{\algorithmiccomment}[1]{$\triangleright$ #1}

\makeatletter
\algnewcommand{\StructComment}[1]{$\triangleright$ \{#1\}}
\makeatother

\makeatletter
\algnewcommand{\LineComment}[1]{\State \(\triangleright\) #1}
\makeatother

\algblock[Name]{StructBegin}{StructEnd}

\algblockdefx[NAME]{START}{END}%
[2][Unknown]{\textbf{StructStart} #1[#2]}%
{\textbf{StructEnd}}

\makeatletter
\renewcommand{\ALG@beginalgorithmic}{\sffamily}
\makeatother

\makeatletter
\def\therule{\makebox[\algorithmicindent][l]{\hspace*{.5em}\vrule height .75\baselineskip depth .25\baselineskip}}%

\newtoks\therules
\therules={}
\def\appendto#1#2{\expandafter#1\expandafter{\the#1#2}}
\def\gobblefirst#1{
  #1\expandafter\expandafter\expandafter{\expandafter\@gobble\the#1}}%
\def\LState{\State\unskip\the\therules}

\def\pushindent{\appendto\therules\therule}%
\def\popindent{\gobblefirst\therules}%
\def\printindent{\unskip\the\therules}%
\def\printandpush{\printindent\pushindent}%
\def\popandprint{\popindent\printindent}%

\algdef{SE}[WHILE]{While}{EndWhile}[1]
  {\printandpush\algorithmicwhile\ #1\ \algorithmicdo}
  {\popandprint\algorithmicend\ \algorithmicwhile}%
\algdef{SE}[FOR]{For}{EndFor}[1]
  {\printandpush\algorithmicfor\ #1\ \algorithmicdo}
  {\popandprint\algorithmicend\ \algorithmicfor}%
\algdef{SE}[FOR]{ForAll}{EndFor}[1]
  {\printandpush\algorithmicfor\ #1\ \algorithmicdo}
  {\popandprint\algorithmicend\ \algorithmicfor}%
\algdef{SE}[LOOP]{Loop}{EndLoop}
  {\printandpush\algorithmicloop}
  {\popandprint\algorithmicend\ \algorithmicloop}%
\algdef{SE}[REPEAT]{Repeat}{Until}
  {\printandpush\algorithmicrepeat}[1]
  {\popandprint\algorithmicuntil\ #1}%
\algdef{SE}[IF]{If}{EndIf}[1]
  {\printandpush\algorithmicif\ #1\ \algorithmicthen}
  {\popandprint\algorithmicend\ \algorithmicif}%
\algdef{C}[IF]{IF}{ElsIf}[1]
  {\popandprint\pushindent\algorithmicelse\ \algorithmicif\ #1\ \algorithmicthen}%
\algdef{Ce}[ELSE]{IF}{Else}{EndIf}
  {\popandprint\pushindent\algorithmicelse}%
\algdef{SE}[FUNCTION]{Function}{EndFunction}[2]
   {\printandpush\algorithmicfunction\ \textproc{#1}\ifthenelse{\equal{#2}{}}{}{(#2)}}%
   {\popandprint\algorithmicend\ \algorithmicfunction}%
\makeatother

\algrenewcommand\algorithmicindent{1.0em}

\algnewcommand\algorithmicprint{\textbf{print }}
\algnewcommand{\Print}[1]{\algorithmicprint #1}

\algnewcommand\algorithmicand{\textbf{and}}
\algnewcommand{\AND}{\algorithmicand}

\algnewcommand\algorithmicswitch{\textbf{switch}}
\algnewcommand\algorithmiccase{\textbf{case}}
\algnewcommand\algorithmicassert{\texttt{assert}}
\algnewcommand\Assert[1]{\State \algorithmicassert(#1)}%
\algdef{SE}[SWITCH]{Switch}{EndSwitch}[1]{\algorithmicswitch\ #1\ \algorithmicdo}{\algorithmicend\ \algorithmicswitch}%
\algdef{SE}[CASE]{Case}{EndCase}[1]{\algorithmiccase\ #1}{\algorithmicend\ \algorithmiccase}%
\algtext*{EndSwitch}%
\algtext*{EndCase}%

\algrenewcommand\Return{\State \algorithmicreturn{} }%

\newcommand{\HCC}{\mathcal{H}}

\newcommand{\KCC}{\mathcal{K}}

\newcommand{\MCC}{\mathcal{M}}
\newcommand{\NCC}{\mathcal{N}}

\newcommand{\VCC}{\mathcal{V}}

\newcommand{\HKCC}{\mathcal{H_K}}
\newcommand{\OOHKCC}{\overline{\mathcal{H_K}}}

\newcommand{\BBR}{\mathbb{R}}

\usepackage[title,titletoc,toc]{appendix}
\usepackage{cleveref}

\crefname{appsec}{Appendix}{Appendices}

\makeatother

\addto\captionsbritish{\renewcommand{\algorithmname}{Algorithm}}
\addto\captionsenglish{\renewcommand{\algorithmname}{Algorithm}}

\begin{document}

\title{Approximation Algorithms for Max-Morse Matching}

\author[1]{Abhishek Rathod\thanks{abhishek@jcrathod.in}} 
\author[1]{Talha Bin Masood\thanks{tbmasood@csa.iisc.ernet.in}} 
\author[1]{Vijay Natarajan\thanks{vijayn@csa.iisc.ernet.in}} 
\affil[1]{Department of Computer Science and Automation, Indian Institute of Science, Bangalore, India} 
\renewcommand\Authands{ and }

\maketitle 

\begin{abstract}
In this paper, we prove that the Max-Morse Matching Problem is approximable, thus resolving an open problem posed by Joswig and Pfetsch~\cite{JP06}. For $D$-dimensional simplicial complexes, we obtain a $\nicefrac{(D+1)}{(D^2+D+1)}$-factor approximation ratio using a simple edge reorientation algorithm that removes cycles. We also describe a $\nicefrac{2}{D}$-factor approximation algorithm for simplicial manifolds by processing the simplices in increasing order of dimension. This algorithm may also be applied to non-manifolds resulting in a $\nicefrac{1}{(D+1)}$-factor approximation ratio. One application of these algorithms is towards efficient homology computation of simplicial complexes. Experiments using a prototype implementation on several datasets indicate that the algorithm computes near optimal results.  
\end{abstract}


\section{Introduction \label{sec:Introduction}}

Discrete Morse theory is a combinatorial analogue of Morse theory
that is applicable to cell complexes~\cite{Fo98}. It has become a
popular tool in computational topology and visualization communities~\cite{Caz03,shiva2012}
and is actively studied in algebraic, geometric, and topological combinatorics~\cite{Koz07,Mil07}. 

The idea of using discrete Morse theory to speedup homology~\cite{HMMNWJD10},
persistent homology~\cite{MN13} and multidimensional persistence~\cite{Lan16}
computations hinges on the fact that discrete Morse theory helps reduce
the problem of computing homology groups on an input simplicial complex
to computing homology groups on a collapsed cell complex. Ideally,
if one were to compute a discrete gradient vector field with minimum
number of critical simplices (unmatched vertices in the Hasse graph)
or maximum number of regular simplices (matched Hasse graph vertices),
then the time required for computing homology over the collapsed cell
complex would be the smallest. However, finding a vector field with
maximum number of gradient pairs is an NP-hard problem as observed
by Lewiner~\cite{Le02} and Joswig et.al.~\cite{JP06} by showing
a reduction from \emph{the collapsibility problem} introduced by E\v gecio\v glu
and Gonzalez in \cite{EG96}. We study the problem of efficiently
computing an approximation to the maximum number of gradient pairs
in a discrete gradient vector field. 

Computing the homology groups has several applications, particularly,
in material sciences, imaging, pattern classification and computer
assisted proofs in dynamics~\cite{KMM10}. More recently, homology
and persistent homology have been appraised to be a more widely applicable
computational invariant of topological spaces, arising from practical
data sets of interest~\cite{Ca09}. An approximately optimal Morse
matching computed using the algorithms described in this paper may
be used towards efficient computation of homology. One of the primary
motivations for us to initiate the study of approximation algorithms
for discrete Morse theory was that a previous study~\cite{HMMNWJD10}
involving discrete Morse theory in homology computation reported noteworthy
speedup over existing methods. Their method used a modification of
the coreduction heuristic~\cite{MB09} to construct discrete Morse
functions. We start with a twin goal in mind~--~first to introduce
rigour into the study by developing algorithms with approximation
bounds and then to have a practical implementation that achieves nearly
optimal solutions.

\subsection{Max Morse Matching Problem}

The Max Morse Matching Problem (\textsf{MMMP}) can be described as
follows: Consider the Hasse graph $\mathcal{H_{K}}$ of a simplicial
complex $\mathcal{K}$ whose edges are all directed from a simplex
to its lower dimensional facets. Associate a matching induced reorientation
to $\mathcal{H_{K}}$ such that the resulting oriented graph $\overline{\mathcal{H_{K}}}$
is acyclic. The goal is to maximize the cardinality of matched (regular)
nodes. Equivalently, the goal is to maximize the number of gradient
pairs. The approximate version of Max Morse Matching Problem seeks
an algorithm that computes a Morse Matching whose cardinality is within
a factor $\alpha$ of the optimal solution for every instance of the
problem.

\subsection{Prior work}

Joswig et al.~\cite{JP06} established the NP-completeness of Morse
Matching Problem. They also posed the approximability of Max Morse
Matching as an open problem  pg.~6~Sec.~4~\cite{JP06}. Several followup
efforts seek optimality of Morse matchings either by restricting the
problem to $2$-manifolds or by applying heuristics~\cite{AFV12,BLW11,HMMNWJD10,He05,JP06,Le03b,Le03a}.
Recently, Burton et al.~\cite{BLPS13} developed an FPT algorithm
for designing optimal Morse functions.

\subsection{Summary of results}

We describe a $\ensuremath{\nicefrac{(D+1)}{(D^{2}+D+1)}}$-factor
approximation algorithm for Max Morse Matching Problem on $D$-dimensional
simplicial complexes. This algorithm uses maximum cardinality bipartite
matching on the Hasse graph $\mathcal{H_{K}}$ to orient it. We then
use a BFS-like traversal of the oriented Hasse graph $\overline{\mathcal{H_{K}}}$
to classify matching edges as either forward edges if they do not
introduce cycles or backward edges if they do. We then use a counting
argument to prove an approximation bound that holds for manifold as
well as non-manifold complexes.

For simplicial manifolds, we propose two approximation algorithms
that exploit the multipartite structure of the Hasse graph. The first
approximation algorithm provides a ratio of $\nicefrac{2}{(D+1)}$.
The ratio is improved to $\nicefrac{2}{D}$ via a refinement that
specifies the order in which the graph is processed. Both algorithms
process simplices of lowest dimension first and then move onto increasingly
higher dimensions. Every $d$-dimensional simplex is first given
the opportunity to match to a $(d-1)$-dimensional simplex. If unsuccessful,
it is then given the option of matching to a $(d+1)$-dimensional
simplex. Furthermore, both algorithms employ optimal algorithms for
designing gradient fields for $0$-dimensional and $D$-dimensional
simplices (in case of manifolds). The refinement processes subgraphs
with small vertex degree with higher priority and hence achieves the
better approximation ratio.

We provide evidence of practical utility of our algorithms through
an extensive series of computational experiments.

\section{Background \label{sec:Background}}

\subsection{Discrete Morse theory \label{sub:Discrete-Morse-Theory}}

Our focus in this paper is limited to simplicial complexes and hence
we restrict the discussion of Forman's Morse theory below to simplicial
complexes. Please refer to~\cite{Fo02} for a compelling expository
introduction. 

Let $\KCC$ be a simplicial complex and let $\sigma^{d},\tau^{d-1}$
be simplices\footnote{An d-dimensional simplex $\sigma^{d}$ may be denoted either as $\sigma$ or $\sigma^{d}$ depending on whether we wish to emphasize its dimension.}
of $\KCC$. The relation $\prec$ is defined as: $\tau\prec\sigma\Leftrightarrow\{\tau\subset\sigma\,\,\textnormal{and}\,\,\dim\tau=\dim\sigma-1\}$.
Alternatively, we say that $\tau$ is the \emph{facet} of $\sigma$
and $\sigma$ is a \emph{cofacet} of $\tau$. The boundary $bd(\sigma)$
and the coboundary $cbd(\sigma)$ of a simplex are defined as: $bd(\sigma)=\{\tau|\tau\prec\sigma\}$
and $cbd(\sigma)=\{\rho|\sigma\prec\rho\}$. A function  $f:\KCC\rightarrow\BBR$
is called a \emph{discrete Morse function} if it assigns higher values
to cofacets, with at most one exception at each simplex. Specifically,
a function $f:\KCC\rightarrow\BBR$ is a discrete Morse function if
for every $\sigma\in\KCC$, $\NCC_{1}(\sigma)=\#\{\rho\in cbd(\sigma)|f(\rho)\leq f(\sigma)\}\leq1$
and $\NCC_{2}(\sigma)=\#\{\tau\in bd(\sigma)|f(\tau)\geq f(\sigma)\}\leq1$.
If $\NCC_{1}(\sigma)=\NCC_{2}(\sigma)=0$ then the simplex $\sigma$
is \emph{critical}, else it is \emph{regular}. 

A pair of simplices $\left\langle \alpha^{m},\beta^{(m+1)}\right\rangle $
with $\alpha\prec\beta$ and $f(\alpha)\geq f(\beta)$ determines
a \emph{gradient pair}. Each simplex must occur in at most one gradient
pair of $\VCC$. A \emph{discrete gradient vector field} $\VCC$ corresponding
to a discrete Morse function $f$ is a collection of simplicial pairs
$\left\langle \alpha^{(p)},\beta^{(p+1)}\right\rangle $ such that
$\left\langle \alpha^{(p)},\beta^{(p+1)}\right\rangle \in\VCC$ if
and only if $f(\alpha)\ge f(\beta)$. 

A simplicial sequence $\{\sigma_{0}^{m}\!\!,\tau_{0}^{m+1}\!\!,\sigma_{1}^{m}\!\!,\tau_{1}^{m+1}\!\!,\dots\sigma_{q}^{m}\!\!,\tau_{q}^{m+1}\!\!,\sigma_{q+1}^{m}\}$
consisting of distinct simplices $(\sigma_{i}\prec\tau_{i})\in\mathcal{V}$
and $\sigma_{i+1}\prec\tau_{i}$ is called a \emph{gradient path}
of $f$.

\subsection{The Hasse graph of a simplicial complex\label{sub:The-Hasse-Graph}}

The \emph{Hasse graph} $\mathcal{H_{K}}$ of a simplicial complex
$\KCC$ is an undirected graph whose vertices are in one-to-one
correspondence with the simplices of the complex. To every simplex
$\sigma_{\KCC}^{d}\in\KCC$ associate a vertex $\sigma_{\HCC}^{d}\in\HKCC$.
The edges in the Hasse graph are determined by facet incidences. $\HKCC$
contains an edge between a vertex that represents\footnote{From here on, for the sake of brevity, while referring to the vertex in $\HKCC$ representing simplex $\sigma^{d}$, we drop the suffix $\mathcal{H}$ from $\sigma_{\HCC}^{d}$. i.e. Instead of referring to it as  vertex $\sigma_{\HCC}^{d}$ we refer to it as simplex $\sigma^{d}$.}
simplex $\sigma^{d}$ and a vertex that represent simplex $\tau^{d-1}$
if and only if $\tau\prec\sigma$. 

We refer to the set of vertices in $\HKCC$ representing $d$-dimensional
simplices as the $d$\emph{-level} of the Hasse graph. The $d$\emph{\emph{-}interface}
of $\HKCC$ is the subgraph consisting of vertices in the $d$-level
and the $(d-1)$-level of of $\HKCC$ together with all the edges
connecting these two levels.

The Hasse graph $\mathcal{H_{K}}$ of a $D$-dimensional complex
$\mathcal{K}$ has $(D+1)$ levels and $D$ interfaces. To understand
Morse theory in terms of Hasse graph, one needs to assign orientations
to it. 

\begin{defn}[Oriented Hasse Graph, Up-edges, Down-edges] If we assign orientations to all edges of Hasse graph $\ensuremath{\HKCC}$, we obtain an \emph{oriented Hasse graph} denoted by $\OOHKCC$. In graph $\OOHKCC$, for any two simplices $\sigma^{d},\tau^{d-1}$, an edge $\tau^{d-1}\rightarrow\sigma^{d}\in\OOHKCC$ going from lower dimensional simplex $\tau^{d-1}$ to a higher dimensional simplex $\sigma^{d}$ is called an \emph{up-edge}. The edge $\sigma^{d}\rightarrow\tau^{d-1}\in\OOHKCC$ is called a \emph{down-edge}.
\end{defn}

If we orient $\ensuremath{\HKCC}$ in such a way that all edges are
down-edges then this orientation corresponds to the trivial gradient
vector field on complex $\mathcal{K}$, for which all simplices are
critical. We call this the \emph{default orientation} on $\HKCC$. 

\textbf{Matching based reorientation.} Start with the default orientation
on $\HKCC$. Associate a matching $\MCC$ to $\HKCC$. If an edge
$\langle\tau^{d-1},\sigma^{d}\rangle\in\MCC$ then \emph{reverse}
the orientation of that edge to $\tau^{d-1}\rightarrow\sigma^{d}\in\OOHKCC$.
We require the matching induced reorientation to be such that the
graph $\overline{\HCC_{\KCC}}$ is a directed acyclic graph. Chari~\cite{Ch00}
first observed that every matching based orientation of $\HKCC$ that
leaves the graph $\overline{\HKCC}$ acyclic corresponds to a unique
gradient vector field on complex $\mathcal{K}$. For such a matching
based acyclic orientation of the graph, every up-edge in the oriented
Hasse graph corresponds to a gradient pair and every unmatched vertex
corresponds to a critical simplex of the gradient vector field. Not
every matching based orientation of $\HKCC$ will leave $\overline{\HKCC}$
acyclic. Figure~\ref{fig:HasseGraph} shows a simplicial complex and
a matching based reorientation of the Hasse graph. We can now define
the Max Morse Matching Problem more formally. 

\begin{figure}
\centering

\includegraphics[width=12cm]{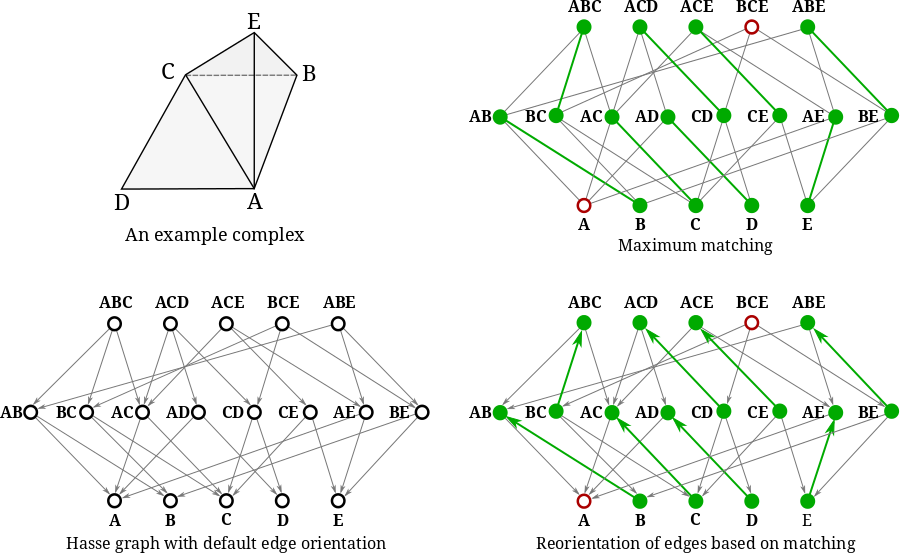}

\caption{Consider a simplicial complex with five triangles (ABC, ACD, ACE,
BCE, ABE). We obtain the oriented Hasse graph for a simplicial complex
(left) and its matching induced orientation (right). Two simplices
are critical (hollow) and others are regular (filled). }

\label{fig:HasseGraph}
\end{figure}

\begin{defn}[Max Morse Matching Problem] A discrete gradient vector field that maximizes the number of gradient pairs over the set of all discrete gradient vector fields on a simplicial complex $\KCC$ is known as a Maximum Morse Matching on $\KCC$. The Max Morse Matching Problem is to find such an  optimal Morse Matching. In terms of Hasse graph, the Max Morse Matching Problem may be defined as the the maximum cardinality of an acyclic matching.  
\end{defn}

We now discuss a few properties of cycles and paths in a matching
based orientation of $\HKCC$. Matching based orientations have the
interesting property that all cycles are restricted to a fixed interface
in the oriented Hasse graph. In other words, if a cycle were to span
multiple interfaces in the Hasse graph, then it will violate the condition
that the orientation is matching based. Similarly, all edges in a
given path belong to a unique interface of the Hasse graph. Also,
in a matching based orientation, source nodes and sink nodes in the
$d$-interface are not involved in any cycles in the $d$-interface.

\begin{defn}[Source and Sink Nodes] A simplex $\sigma^{d}$ is a \emph{source node} for the $d$-interface if it has only outgoing edges to $d-1$ simplices. If in addition, simplex $\sigma^{d}$ is matched to a $(d+1)$-simplex then it is as a \emph{regular source node} for the $d$-interface, else it is a \emph{critical source node}. Similarly, a simplex $\tau^{d-1}$ is a \emph{sink node} for the $d$-interface if it has only incoming edges from $d$ simplices. If $\tau^{d-1}$ is matched to a $(d-2)$-simplex then it is known as a \emph{regular sink node} else it is known as a \emph{critical sink node}.   
\end{defn}

\section{A $\ensuremath{\nicefrac{(D+1)}{(D^{2}+D+1)}}$-factor approximation
algorithm for simplicial complexes \label{sec:A--factor-Approximation}}

We now describe an approximation algorithm for the Max Morse Matching
Problem that is applicable to simplicial complexes. The idea is to
first compute a maximum cardinality matching and, in a subsequent
step, remove any cycles that maybe introduced due to the reorientation.
The key steps are outlined in Algorithm~\ref{alg:frontier}. We begin
with notes on notations and definitions. 

\textbf{Notation.} When we denote an up-edge as $\chi(\alpha,\beta)$,
we mean to say that it is an edge connecting simplex $\alpha^{d-1}$
to simplex $\beta^{d}$ and is labelled as $\chi$. We may write it
either as $\chi(\alpha,\beta)$ or $\chi$ depending on whether we
want to emphasize vertices incident on $\chi$. The corresponding
down-edge with reversed orientation is denoted as $\overline{\chi}$
or $\overline{\chi}(\beta,\alpha)$.

\begin{defn}[Leading up-edges of an up-edge] In an oriented Hasse graph $\OOHKCC$, if we have an up-edge $\chi_1(\alpha_1,\beta_1)$ followed by a down-edge $\overline{\chi_2}(\beta_1,\alpha_2)$ followed by up-edge $\chi_3(\alpha_2,\beta_2)$ we say that $\chi_3$ is a leading up-edge of $\chi_1$. 
\end{defn}

\begin{defn}[Facet-edges of a simplex] In an oriented Hasse graph $\OOHKCC$, for a simplex $\sigma^{d}$ (where $d\geq 1$), the set of oriented edges between $\sigma^{d}$ to $(d-1)$-simplices incident on $\sigma^{d}$ (along with respective orientations) are known as the facet-edges of $\sigma^{d}$.     
\end{defn} 

\begin{algorithm}[h]
\caption{The Frontier Edges Algorithm}

\begin{algorithmic}[1]
\Require{Simplicial complex $\mathcal{K}$}
\Ensure{Graph $\mathcal{H_V}$, an acyclic matching based orientation of Hasse graph $\HCC_\KCC$ of $\KCC$.}
\LState{Construct Hasse graph $\mathcal{H_{K}}$ of $\KCC$.}
\LState{Perform maximum-cardinality graph matching on $\mathcal{H_{K}}$.}
\LState{Let $\OOHKCC$ denote the matching induced reorientation of $\HKCC$ and $\mathcal{E}(\OOHKCC)$ its edge set.}
\LState{Initialize the edge set of $\HCC_\VCC$, $\mathcal{E}(\mathcal{H_{V}})\leftarrow\emptyset$.}
\While{$\exists\chi\in\mathcal{E}(\OOHKCC)$ such that $\chi$ is an up-edge}
	\LState{$\mathcal{C}_{\chi}\leftarrow$\textbf{BFSComponent}($\OOHKCC,\chi$)}
	\LState{$\mathcal{E}(\OOHKCC)\leftarrow\mathcal{E}(\OOHKCC)\setminus\mathcal{C}_{\chi}$}
	\LState{$\mathcal{E}(\mathcal{H_{V}})\leftarrow\mathcal{E}(\mathcal{H_{V}})\cup\mathcal{C}_{\chi}$}
\EndWhile  
\LState{$\mathcal{E}(\mathcal{H_{V}})\leftarrow\mathcal{E}(\mathcal{H_{V}})\cup\mathcal{E}(\OOHKCC)$}\label{lst:line:regsource}

\vspace{1.5mm}

\Procedure{BFSComponent}{$\OOHKCC,\chi$}
	\LState{$\mathcal{C}\leftarrow\emptyset$}
	\LState{Initialize the queue $\mathcal{Q}\leftarrow\emptyset$}
	\LState{$\textbf{enqueue}(\mathcal{Q},\chi)$}
	\While{$\mathcal{Q}$ is non-empty}
		\LState{$\chi_{0}(\alpha_{0},\beta_{0})\leftarrow \textbf{dequeue}(\mathcal{Q})$}
		\LState{$\mathcal{C}\leftarrow\mathcal{C}\cup\textbf{facetEdges}(\beta_{0})$} 
		\For{every leading up-edge $\chi_{i}(\alpha_{i},\beta_{i})$  of $\chi_{0}$}
			\If{the graph induced by edges in $(\mathcal{C} \cup \textbf{facetEdges}(\beta_{i}))$ has cycles }
				\LState{Reverse orientation of $\chi_{i}$ in graph $\OOHKCC$}
				\LState{$\mathcal{C}\leftarrow\mathcal{C}\cup\textbf{facetEdges}(\beta_{i})$}
			\Else
				\LState{\textbf{enqueue}($\mathcal{Q},\chi_{i}$)}
			\EndIf
		\EndFor
	\EndWhile
	\LState{\textbf{return} $\mathcal{C}$}
\EndProcedure
\end{algorithmic}

\label{alg:frontier}
\end{algorithm}

Given a Hasse graph $\mathcal{H_{K}}$ on complex $\mathcal{K}$,
Algorithm~\ref{alg:frontier} begins by computing maximum cardinality
graph matching on graph $\mathcal{H_{K}}$ and then uses this matching
to induce an orientation on $\mathcal{H_{K}}$. Let $\OOHKCC$ denote
the oriented Hasse graph based on graph matching and $\mathcal{H_{V}}$
denote the output graph. While there exists an up-edge $\chi$ in
$\OOHKCC$, we make $\chi$ a \emph{seed\emph{-}edge} and use it
as a starting point for a BFS-like traversal on graph $\OOHKCC$.
This traversal is done using procedure \textbf{BFSComponent}() which
returns a set of edges $\mathcal{C}_{\chi}$. The \emph{edge\emph{-}component
$\mathcal{C}_{\chi}$} of a seed edge $\chi$ is the set of edges
discovered in the BFS-like traversal of graph $\OOHKCC$, with $\chi$
as the start edge. Each time, we discover a new edge-component, we
delete it from $\OOHKCC$ and add it to $\mathcal{H_{V}}$. We exit
the while loop when all up-edges are exhausted. 

If a simplex $\sigma^{d}$ is either a critical node or a regular
source node, then its facet-edges are not reachable in the BFS traversal
through any of the up-edges in $\OOHKCC$. In a final step, we include
all remaining edges from $\OOHKCC$ to $\mathcal{H_{V}}$.

The procedure \textbf{BFSComponent}() computes the component edges
by processing edges from the queue one at a time. Let $\chi_{0}(\alpha_{0},\beta_{0})$
be the edge at the top of the queue. We add all the facet-edges of
simplex $\beta_{0}$ to the edge-component $\mathcal{C}$. We now
examine the leading up-edges of $\chi_{0}$. If $\chi_{i}(\alpha_{i},\beta_{i})$
is a leading up-edge of $\chi_{0}$ then we check if the addition
of facet-edges of simplex $\beta_{i}$ to $\mathcal{C}$ creates
cycles. If it does then we classify $\chi_{i}$ as a \emph{backward
edge}, reverse the orientation of $\chi_{i}$ and add the facet-edges
of $\beta_{i}$ to $\mathcal{C}$. If this addition does not introduce
cycles, then we classify $\chi_{i}$ as a \emph{forward edge} and
enqueue it in the queue of up-edges. Please refer to Figure~\ref{fig:forwardbackward}.
Enqueuing $\chi_{i}$ guarantees that at some stage when $\chi_{i}$
gets dequeued, we will end up adding facet-edges of simplex $\beta_{i}$
to $\mathcal{C}$. When the queue is exhausted, $\mathcal{C}$ contains
the entire edge-component of some seed-edge.

\selectlanguage{english}%
\begin{figure}
\centering

\includegraphics[width=8cm]{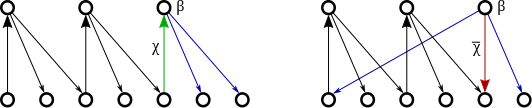}

\caption{Two cases in BFSComponent(). Left: A forward edge $\chi$ is identified.
The edge $\chi$ along with the down-edges incident on $\beta$ are
added to the edge\selectlanguage{american}%
-\selectlanguage{english}%
component. Right: A backward edge $\chi$ is identified. The edge
$\bar{\chi}$ along with the down-edges incident on $\beta$ are
added to the edge-component. }

\label{fig:forwardbackward}
\end{figure}

\selectlanguage{american}%
We first prove an acyclicity lemma on edge\selectlanguage{english}%
-\selectlanguage{american}%
components returned by procedure \textbf{BFSComponents}()\textbf{
}in Algorithm 1\textbf{.}

\begin{lem} \label{lem:nocycomp} The graph induced by edges in an edge-component is acyclic.
\end{lem}
\begin{proof} 
Consider the graph induced by edges in edge-component $\mathcal{C}$ belonging to a $d$-interface.
We know that an up-edge say $\chi_j$ is classified as a forward edge if and only if the inclusion of $\chi_j$ does not create a cycle with up-edges that were included prior to $\chi_j$ in edge-component $\mathcal{C}$. Hence, we can be sure that inclusion of set of all forward edges does not create cycles. Moreover, every time a backward edge, say $\chi_i(\alpha_i,\beta_i)$ is encountered, we include the inverse orientation of $\chi_i$ in $\mathcal{C}$ which creates a sink node at $\alpha_i$ and source node at $\beta_i$ for the $d$-interface of the Hasse graph. Also, the $(d-1)$-simplices that were \emph{visited} in a previous edge-component also act as sinks (since we restrict ourselves to edges induced by edge-component $\mathcal{C}$). Furthermore, every down-edge is incident on a $(d-1)$-simplex that is either a sink or a $(d-1)$-simplex incident on a forward edge. In either case, it is easy to see that all flow terminates at sinks making the graph induced by edges in a particular edge-component acyclic.
\end{proof}

\begin{lem}\label{lem:acyclic1} The output graph $\mathcal{H_{V}}$ is acyclic.
\end{lem}
\begin{proof} We prove this claim via induction over sequential addition of edge-components.\\
Base Case: To begin with the output graph $\mathcal{H_{V}}$ is the empty graph. From Lemma~\ref{lem:nocycomp}, we know that the graph induced by edges in an edge-component is acyclic. So $\mathcal{H_{V}}$ remains acyclic following the addition of the first edge-component $\mathcal{C}_{1}$ to $\mathcal{H_{V}}$.\\
Inductive Hypothesis: Suppose that following the addition of edges belonging to $i^{th}$ edge-component $\mathcal{C}_{i}$, $\mathcal{H_{V}}$ remains acyclic.\\
Now, we need to prove that following the addition of edges belonging to $\mathcal{C}_{i+1}$, $\mathcal{H_{V}}$ remains acyclic. To begin with using Lemma~\ref{lem:nocycomp}, we note that the graph induced by $\mathcal{C}_{i+1}$ is acyclic. So, if there does exist a cycle in $\mathcal{H_{V}}$ following the addition of $\mathcal{C}_{i+1}$, then a forward up-edge of this cycle must belong to $\mathcal{C}_{i+1}$ and a forward up-edge must belong to an edge-component $\mathcal{C}_{j_k}$ where $j_k<(i+1)$.
In particular, this means that there exists a down-edge belonging to a component $\mathcal{C}_{j_k}$ that is incident on simplex $\alpha_1$ such that a forward edge $\chi_1(\alpha_1,\beta_1)\in\mathcal{C}_{i+1}$. But, if $\alpha_1$ was reachable while traversing $\mathcal{C}_{j_k}$ then $\chi_1(\alpha_1,\beta_1)$ would have been classified as a forward edge in $\mathcal{C}_{j_k}$ i.e. $\chi_1(\alpha_1,\beta_1)\in\mathcal{C}_{j_k}$  -- a contradiction. Hence, such cycles do not exist. Finally, in line~\ref{lst:line:regsource} of Algorithm~\ref{alg:frontier}, after having added all edge-components, we add all the facet-edges of $d$-simplices that are either unmatched or facet-edges of $d$-simplices that are matched to one of their cofacets. In such cases, they act as source nodes within $d$-interfaces and do not introduce cycles because all cycles are restricted to the $d$-interface. 
\end{proof} 

\begin{lem}\label{lem:orient1} The output graph $\mathcal{H_{V}}$ is a matching based acylic orientation of undirected Hasse graph of the complex $\mathcal{H_K}$. 
\end{lem}
\begin{proof}  We first prove that $\mathcal{H_{V}}$ is an orientation of $\mathcal{H_K}$ i.e. for every undirected edge in $\mathcal{H_K}$ there is a corresponding directed edge in $\mathcal{H_{V}}$. To prove this we will show that for every simplex $\beta^{d}$, all undirected edges from $\beta$ to its facets in $\mathcal{H_K}$ has a corresponding oriented edge in $\mathcal{H_{V}}$.\\
Case 1: Suppose that $\beta$ is matched to one of its facets in max-cardinality matching induced oriented graph $\OOHKCC$. Then this up-edge incident on $\beta$ was classified either as a forward edge or as a backward edge. In either case, all its facet-edges are inserted in $\mathcal{H_{V}}$ in procedure \textbf{BFSComponent}(). \\
Case 2: Now suppose that $\beta$ is either unmatched or it is matched to one of its cofacets. Then clearly, none of its facet-edges can be reached through a graph traversal that starts with some up-edge in $\OOHKCC$. Therefore, these facet-edges are not a part of any of the edge-components and they are all down-edges. However, in line~\ref{lst:line:regsource} of Algorithm~\ref{alg:frontier}, all these \emph{remainder} edges are included in $\mathcal{H_{V}}$. \\
Since the above two cases hold true for every simplex $\sigma^{d}$ with $d\geq 1$, this proves that $\mathcal{H_{V}}$ is an orientation of graph $\mathcal{H_K}$. Also, given the fact that the up-edges that are included are subset of edges coming from cardinality bipartite matching, clearly the orientation of $\mathcal{H_{V}}$ is matching based. In Lemma~\ref{lem:acyclic1}, we already proved that graph $\mathcal{H_{V}}$ is acyclic. Hence proved.
\end{proof}

\begin{defn}[Classified Edges, Frontier Edges] An edge marked within the \textbf{BFSComponent}() as forward or backward is called a \emph{classified edge}. A leading up-edge that is not yet classified is called a \emph{frontier edge}.  
\end{defn}

We establish the approximation ratio using a counting argument that
works specifically for simplicial complexes. We refer to this argument
as the frontier edges argument. The main idea involves a method of
counting that we describe now. Suppose we are processing an edge-component
that belongs to the $d$-interface of the Hasse graph for some $d\leq D$.
Let the iterator variable $i$ count the number of up-edges in the
edge-component that have so far been classified as either forward
or backward. Suppose at the end of ${i}^{th}$ iteration, there are
$|\mathcal{F}_{i}|$ number of forward edges, $|\mathcal{B}_{i}|$
number of backward edges and $|\mathcal{Z}_{i}|$ number of frontier
edges, then our approximation ratio is calculated as $\nicefrac{|\mathcal{F}_{i}|}{(|\mathcal{F}_{i}|+|\mathcal{B}_{i}|+|\mathcal{Z}_{i}|)}$.
In other words, we assume the worst case scenario where all the frontier
edges are \emph{possibly backward}. In every iteration of the BFS,
we classify one of the frontier edges as a forward edge or a backward
edge and then update the ratio until we exhaust the entire edge-component.
In the ${(i+1)}^{th}$ iteration, if a frontier edge is classified
as a forward edge then the number of forward edges will be $|\mathcal{F}_{i+1}|=(|\mathcal{F}_{i}|+1)$
and the number of frontier edges will be $|\mathcal{Z}_{i+1}|=(|\mathcal{Z}_{i}|+d-1)$.
If a frontier edge is classified as a backward edge then the number
of backward edges will be $|\mathcal{B}_{i+1}|=(|\mathcal{B}_{i}|+1)$
and the number of frontier edges will be $|\mathcal{Z}_{i+1}|=(|\mathcal{Z}_{i}|-1)$. 

\selectlanguage{english}%
\begin{figure}
\centering

\includegraphics[width=12cm]{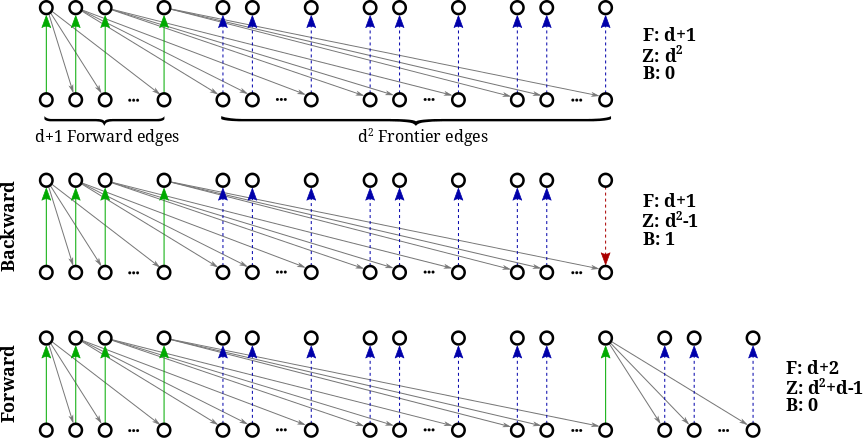}

\caption{The frontier edges argument. Top: The base case with $d+1$ forward
edges, $d^{2}$ frontier edges and no backward edges. The ratio of
forward edges and total number of up-edges is $\frac{d+1}{d^{2}+d+1}$.
Middle: When one of the frontier edges is classified as backward,
the ratio remains the same. Bottom: When one of the frontier edges
is classified as a forward edge, the ratio improves to $\frac{d+2}{d^{2}+2d+1}$.}

\label{fig:frontier}
\end{figure}

\selectlanguage{american}%
\begin{lem} The number of forward edges in an edge-component belonging to the $d$-interface of the Hasse graph is at least $\nicefrac{(d+1)}{(d^2+d+1)}$ fraction of the total number of up-edges in the edge-component.
\end{lem}
\begin{proof} We will use induction to prove our claim. \\ 
\textbf{Base Case:} The seed edge $\chi_0$ of the edge-component is naturally a forward edge. We note that any cycle in the Hasse graph of a simplicial complex has minimum length $6$ and involves at least $3$ up-edges. Since this does not hold for general regular cell complexes, simplicial input is crucial for the proof to work. Cycles do not appear until after two iterations. These two iterations constitute the base case. Therefore, $|\mathcal{F}_1|=1$, $|\mathcal{B}_1|=0$ and $|\mathcal{Z}_1|=0$. Also, the leading up-edges of $\chi_0$ are also forward edges. If $\chi_0$ has no leading up-edges then the edge-component is exhausted and $\nicefrac{|\mathcal{F}_1|}{(|\mathcal{F}_1|+|\mathcal{B}_1|)}=1$. If $\chi_0$ has $K$ leading up-edges, each such edge has, in turn, at most $j_k$ leading up-edges then the total number of forward edges will be $|\mathcal{F}_2|=1+K$, $|\mathcal{B}_2|=0$ and $|\mathcal{Z}_2|= \sum\limits_{k=1}^{K} j_k$. It is easy to check that the worst case for ratio $\nicefrac{|\mathcal{F}_2|}{(|\mathcal{F}_2|+|\mathcal{B}_2|+|\mathcal{Z}_2|)}$ occurs when $K=d$ and $j_k=d$ for each $k$. This gives us the worst case ratio for the quantity $\nicefrac{|\mathcal{F}_2|}{(|\mathcal{F}_2|+|\mathcal{B}_2|+|\mathcal{Z}_2|)}$ to be $\nicefrac{(d+1)}{(d^2+d+1)}$. Please refer to Figure~\ref{fig:frontier} \\ 
\textbf{Induction Step:} Our induction hypothesis says that after $i$ iterations of BFS, the ratio $\nicefrac{|\mathcal{F}_i|}{(|\mathcal{F}_i|+|\mathcal{B}_i|+|\mathcal{Z}_i|)}\geq \nicefrac{(d+1)}{(d^2+d+1)}$. For the ${(i+1)}^{th}$ iteration, suppose one of the frontier edges is classified as a forward edge. Then $|\mathcal{F}_{i+1}|=(|\mathcal{F}_i|+1)$ and $|\mathcal{Z}_{i+1}|\leq(|\mathcal{Z}_i|+d-1)$. Note that $(|\mathcal{Z}_i|+d-1)$ is the worst case estimate for $|\mathcal{Z}_{i+1}|$ assuming that the newly included forward edge has $d$ leading up-edges. Therefore, the numerator of the ratio $\nicefrac{|\mathcal{F}_i|}{(|\mathcal{F}_i|+|\mathcal{B}_i|+|\mathcal{Z}_i|)}$ increases by 1 whereas the denominator increases by $d$. Also we have $\nicefrac{1}{d}>\nicefrac{(d+1)}{(d^2+d+1)}$. Using the elementary fact that if $\frac{A}{B}\geq\frac{E}{F}$ and $\frac{C}{D}\geq\frac{E}{F}$ then $\frac{(A+C)}{(B+D)}>\frac{E}{F}$ for non-negative values of $A, B, C, D, E$ and $F$, we get: 
\begin{eqnarray*} \frac{|\mathcal{F}_{i+1}|}{(|\mathcal{F}_{i+1}|+|\mathcal{B}_{i+1}|+|\mathcal{Z}_{i+1}|)} & \geq & \frac{|\mathcal{F}_{i}|+1}{(|\mathcal{F}_{i}|+|\mathcal{B}_{i}|+|\mathcal{Z}_{i}|)+d}\\  & \geq & \frac{(d+1)+1}{(d^{2}+d+1)+d}\\  & > & \frac{(d+1)}{(d^{2}+d+1)} \end{eqnarray*}
On the other hand if a frontier edge is classified as a backward edge then $|\mathcal{B}_{i+1}|=(|\mathcal{B}_i|+1)$ and $|\mathcal{Z}_{i+1}|=(|\mathcal{Z}_i|-1)$. So, the numerator and the denominator of the ratio $\nicefrac{|\mathcal{F}_{i+1}|}{(|\mathcal{F}_{i+1}|+|\mathcal{B}_{i+1}|+|\mathcal{Z}_{i+1}|)}$ remain unchanged which gives us $\nicefrac{|\mathcal{F}_{i+1}|}{(|\mathcal{F}_{i+1}|+|\mathcal{B}_{i+1}|+|\mathcal{Z}_{i+1}|)}=\nicefrac{|\mathcal{F}_i|}{(|\mathcal{F}_i|+|\mathcal{B}_i|+|\mathcal{Z}_i|)}$. In both cases, the bound holds after $(i+1)$ iterations.  
\end{proof} 

Since every edge-component that belongs to a $d$-interface achieves
a ratio of at least $\nicefrac{(d+1)}{(d^{2}+d+1)}$ edges, if we
sum over all the edge-components we get the ratio $\nicefrac{(d+1)}{(d^{2}+d+1)}$
for that $d$-interface. In other words, we preserve at least $\nicefrac{(d+1)}{(d^{2}+d+1)}$
of the total number of matchings at every $d$-interface. The ratio
$\nicefrac{(d+1)}{(d^{2}+d+1)}$ becomes worse with increasing $d$.
So the worst case ratio is $\nicefrac{(D+1)}{(D^{2}+D+1)}$ where
$D$ is the dimension of the complex. Therefore, we get the following
result on the approximation ratio.

\begin{thm} Algorithm 1 computes a $\nicefrac{(D+1)}{(D^2+D+1)}$-factor approximation for Max Morse Matching Problem on simplicial complexes of dimension $D$. 
\end{thm}
\begin{proof} Let $|\mathcal{M}|$ denote the cardinality of maximum matching. Note that $2|\mathcal{M}|$ is an upper bound on Max Morse Matching i.e optimal number of regular simplices $\leq 2|\mathcal{M}|$. Since we preserve at least $\nicefrac{(D+1)}{(D^{2}+D+1)}$ of these matchings, the number of regular simplices we obtain is  at least $2\frac{(D+1)}{(D^{2}+D+1)}|\mathcal{M}|\geq \frac{(D+1)}{(D^{2}+D+1)}OPT$. Therefore, Algorithm 1 provides a $\frac{D+1}{(D^2+D+1)}$-factor approximation for the Max Morse Matching Problem on simplicial complexes. 
\end{proof}

\subsection{A $\nicefrac{5}{11}$-factor Approximation for $2$-dimensional
simplicial complexes using Frontier Edges Algorithm}

In this section, we observe that we can further tighten our analysis
of Algorithm~\ref{alg:frontier} by restricting the problem to 2-dimensional
simplicial complexes. We exploit the geometry of $2$-complexes as
proved in Lemma~\ref{lem:oneonly} in order to establish an improved
ratio in the base case. 

\begin{lem}\label{lem:oneonly} If $\alpha$ is a forward edge and $\beta_1$ is a leading forward edge of edge $\alpha$ and if $\gamma_1$ and $\gamma_2$ are leading up-edges of $\beta_1$ then only one of the two edges $\gamma_1$ and $\gamma_2$ can possibly be a  backward edge that creates a cycle with edge $\alpha$.
\end{lem}
\begin{proof} Without loss of generality, in this proof, we will use concrete labeling of simplices. We make an elementary geometric observation to prove this claim. Suppose $\alpha$ is a forward edge between a $1$-simplex say AB matched to a $2$-simplex ABC. So $\alpha$ can alternatively be denoted as edge AB-ABC. Now suppose $1$-simplex BC is matched to another $2$-simplex BCD constituting forward edge $\beta_1$, then of the two $1$-simplices BD and CD, BD can possibly match a $2$-simplex say BDA which effectively makes edge BD-BDA (say $\gamma_1$) a backward edge. However it is impossible to have a forward edge incident on $1$-simplex CD (say $\gamma_2$) that is also simultaneously incident on $1$-simplex AB since any $2$-simplex has at most three vertices. Hence proved.
\end{proof}

\begin{lem} The number of forward edges is at least $\nicefrac{5}{11}$ fraction of the total number of up-edges in the edge-component.
\end{lem}
\begin{proof} Once again, we will use induction to prove our claim. \\
\textbf{Base Case:} In case of 2-manifolds, we can count up to three levels of BFS for base case, which in turn gives us an improvement in ratio. The seed edge $\alpha$ of the edge-component is evidently a forward edge. We note that any cycle in the Hasse graph of a simplicial complex has minimum length $3$. Therefore, $|\mathcal{F}_1|=1$, $|\mathcal{B}_1|=0$ and $|\mathcal{Z}_1|=0$. Also, the leading up-edges of $\alpha$ (if any) are also forward edges. If $\alpha$ has no leading up-edges then the edge-component is exhausted and $\nicefrac{|\mathcal{F}_1|}{(|\mathcal{F}_1|+|\mathcal{B}_1|)}=1$. If $\alpha$ has one leading up-edge $\beta_1$, then $|\mathcal{F}_2|=2$, $|\mathcal{B}_2|=0$ and $|\mathcal{Z}_2|=2$. Therefore, $\nicefrac{|\mathcal{F}_2|}{(|\mathcal{F}_2|+|\mathcal{B}_2|+|\mathcal{Z}_2|)}=\nicefrac{1}{2}$. If $\alpha$ has two leading up-edges $\beta_1$ and $\beta_2$, then $|\mathcal{F}_2|=3$, $|\mathcal{B}_2|=0$ and $|\mathcal{Z}_2|=4$. Therefore, $\nicefrac{|\mathcal{F}_2|}{(|\mathcal{F}_2|+|\mathcal{B}_2|+|\mathcal{Z}_2|)}=\nicefrac{3}{7}$. By Lemma~\ref{lem:oneonly}, both leading up-edges of $\beta_1$, $\gamma_1$ and $\gamma_2$, can not be backward. So suppose one of them (say $\gamma_1$) is backward and $\gamma_2$ is forward then $|\mathcal{F}_3|=4$, $|\mathcal{B}_3|=1$ and $|\mathcal{Z}_3|=4$ and therefore $\nicefrac{|\mathcal{F}_3|}{(|\mathcal{F}_3|+|\mathcal{B}_3|+|\mathcal{Z}_3|)}=\nicefrac{4}{9}$. Similarly, we must also consider the leading up-edges of $\beta_2$ of which at most one of them can be backward. The worst case occurs for the configuration when exactly one leading up-edge each of $\beta_1$ and $\beta_2$ are backward. This configuration gives $|\mathcal{F}_3|=5$, $|\mathcal{B}_3|=2$ and $|\mathcal{Z}_3|=4$ and hence $\nicefrac{|\mathcal{F}_3|}{(|\mathcal{F}_3|+|\mathcal{B}_3|+|\mathcal{Z}_3|)}=\nicefrac{5}{11}$. \\
\textbf{Induction Step:} Our induction hypothesis is that following $i$ iterations of BFS, the ratio $\nicefrac{|\mathcal{F}_i|}{(|\mathcal{F}_i|+|\mathcal{B}_i|+|\mathcal{Z}_i|)}\geq \nicefrac{5}{11}$. For the ${(i+1)}^{th}$ iteration, suppose one of the frontier edges is classified as a forward edge. Then $|\mathcal{F}_{i+1}|=(|\mathcal{F}_i|+1)$ and $|\mathcal{Z}_{i+1}|=(|\mathcal{Z}_i|+1)$. Therefore, the numerator of the ratio $\nicefrac{|\mathcal{F}_i|}{(|\mathcal{F}_i|+|\mathcal{B}_i|+|\mathcal{Z}_i|)}$ increases by 1 whereas the denominator increases by $2$. However, since $\nicefrac{1}{2}>\nicefrac{5}{11}$, we have $\nicefrac{|\mathcal{F}_{i+1}|}{(|\mathcal{F}_{i+1}|+|\mathcal{B}_{i+1}|+|\mathcal{Z}_{i+1}|)}=\nicefrac{(5+1)}{(11+1)}>\nicefrac{5}{11}$. On the other hand if a frontier edge is classified as a backward edge then $|\mathcal{B}_{i+1}|=(|\mathcal{B}_i|+1)$ and $|\mathcal{Z}_{i+1}|=(|\mathcal{Z}_i|-1)$. So numerator and denominator of ratio $\nicefrac{|\mathcal{F}_{i+1}|}{(|\mathcal{F}_{i+1}|+|\mathcal{B}_{i+1}|+|\mathcal{Z}_{i+1}|)}$ remain unchanged which gives us $\nicefrac{|\mathcal{F}_{i+1}|}{(|\mathcal{F}_{i+1}|+|\mathcal{B}_{i+1}|+|\mathcal{Z}_{i+1}|)}=\nicefrac{|\mathcal{F}_i|}{(|\mathcal{F}_i|+|\mathcal{B}_i|+|\mathcal{Z}_i|)}$. When all up-edges of the edge-component are exhausted, we don't have anymore frontier edges and the ratio for the edge-component after processing $|\mathcal{F}|$ forward edges and $|\mathcal{B}|$ backward edges will be  $\nicefrac{|\mathcal{F}|}{(|\mathcal{F}|+|\mathcal{B}|)}$ and by our inductive argument the ratio will be at least $\nicefrac{5}{11}$.
\end{proof}

Once again since every edge-component achieves a ratio of at least
$\nicefrac{5}{11}$ edges, if we sum over all the edge-components
we get the following theorem as an immediate outcome of the lemma
above.

\begin{thm} Algorithm~\ref{alg:frontier} is a $\nicefrac{5}{11}$-factor approximation algorithm for Max Morse Matching Problem when restricted to 2-dimensional simplicial complexes.
\end{thm}

\section{Approximation algorithms for simplicial manifolds\label{sec:Approximation-Simplicial-Manifolds}}

\subsection{A $\nicefrac{2}{(D+1)}$-factor approximation algorithm for simplicial
manifolds\label{sub:manifold-approx}}

\begin{algorithm}[t]
\caption{The Interface Algorithm}

\begin{algorithmic}[1]
\Require{Simplicial complex $\mathcal{K}$}
\Ensure{Graph $\mathcal{H_V}$, an acyclic matching based orientation of Hasse graph $\HCC_\KCC$ of $\KCC$.}
\LineComment{ \textbf{Notation:} [$\mathcal{C}_{d}^{d-1}$ denotes the critical $(d-1)$-simplices for $d$-interface. $\mathcal{R}_{d}$ is the set of all regular simplices for $d$-interface and $\mathcal{M}_{d}$ is the set of gradient pairs for $d$-interface. $\mathcal{E}(\mathcal{H_{V}})$ denotes the edge set of $\mathcal{H_{V}}$.]}

\LState{Construct Hasse graph $\mathcal{H_{K}}$ of $\KCC$.}
\LState{$\mathcal{E}(\mathcal{H_{V}})$ is initialized to default down-edge orientation on all edges.}
\ForAll{$d=1$ to $D$}
	\LState{$\mathcal{G}^{d}\leftarrow\textbf{extractdInterface}(\mathcal{H_{K}},d)$}
	\LState{ \algorithmicif\ {$d=1$}\ \algorithmicthen\ {Apply $\textbf{1ComplexOpt}(\mathcal{G}^{1})$}} 
	\LState{ \algorithmicelse\ \algorithmicif\ {$d=D$}\ \algorithmicthen\ {Apply $\textbf{manifoldOpt}(\mathcal{G}^{D})$}} 
	\LState{ \algorithmicelse\ {Apply $\textbf{intermediateApx}(\mathcal{G}^{d},d)$}} 
	\LState{\algorithmicend\ \algorithmicif}
\EndFor

\vspace{1.5mm}

\Procedure{deleteAndReorient}{$\mathcal{C}_{d}^{d-1},\mathcal{R}_{d},\mathcal{M}_{d}$}
\LState{Reorient edges of $\mathcal{H_{V}}$ based on matchings in edge set $\mathcal{M}_{d}$ for the $d$-interface)}
\LState{Delete nodes $\left\{ \mathcal{C}_{d}^{d-1},\mathcal{R}_{d}\right\} $ from $\mathcal{H_{K}}$}
\EndProcedure

\vspace{1.5mm}

\Procedure{1ComplexOpt}{$\mathcal{W}$}
\LState{Apply the optimal algorithm on $\mathcal{W}$. (See \textbf{DFSoptimal}() in Appendix~\ref{sec:Optimal-Algorithms-for}).}
\LState{\textbf{deleteAndReorient}($\mathcal{C}_{1}^{0},\mathcal{R}_{1},\mathcal{M}_{1}$)}
\EndProcedure

\vspace{1.5mm}

\Procedure{manifoldOpt}{$\mathcal{W}$}
\LState{Apply the optimal algorithm on the dual graph. (See \textbf{DFSoptimal}() in Appendix~\ref{sec:Optimal-Algorithms-for}).}
\LState{Reorient edges of $\mathcal{H_{V}}$ based on matchings in edge set $\mathcal{M}_{D}$ for the $D$-interface.}
\EndProcedure

\vspace{1.5mm}

\Procedure{intermediateApx}{$\mathcal{W},d$}
\LState{Apply Algorithm 1 described in Section~\ref{sec:A--factor-Approximation} on graph $\mathcal{W}$.} \label{lst:line:canmatch}
\LState{For every unmatched simplex $\tau^{d-1}$ such that all its cofacets $\sigma_{1}^{d}\dots\sigma_{K}^{d}$ are also unmatched, choose one of the simplices $\sigma_{i}^{d},i\in[1,K]$ and introduce the matching $\left\langle \tau,\sigma_{i}\right\rangle$.} \label{lst:line:cantriplet}
\LState{\textbf{deleteAndReorient}($\mathcal{C}_{d}^{d-1},\mathcal{R}_{d},\mathcal{M}_{d}$)}
\EndProcedure
\end{algorithmic}

\label{alg:interface}
\end{algorithm}

We will restrict our attention to manifolds without boundary. The
key idea in Algorithm~\ref{alg:interface} is that the matching is
constructed within one $d$-interface at a time, starting with the
lowest interface and ending with the highest one. For manifolds, this
is advantageous because it allows us to count matched/critical simplices
differently. In particular, every $d$-simplex (where $1\leq d\leq D-1$),
is given two chances to get matched. Please refer to Figure~\ref{fig:interface}.
We first try to match a $d$-simplex say $\sigma^{d}$, while constructing
the Morse matching for the $d$-interface. If $\sigma^{d}$ remains
critical for the $d$-interface then we try to match it for the $(d+1)$-interface.
The trick of giving a second chance to critical simplices works fine
for all dimensions except for $D$-dimensional critical simplices.
Fortunately, for manifolds, we can easily design a vector field with
only one critical simplex for dimension $D$. Since non-manifold-complexes
may have unbounded number of critical $D$-simplices the analysis
becomes non-trivial. For Algorithm~\ref{alg:interface}, one may still
derive approximation bounds for non-manifold complexes by using a
line of reasoning analogous to one used in Section~\ref{sub:apx-non-man-comp}.

\selectlanguage{english}%
\begin{figure}[t]
\centering\includegraphics[width=1\textwidth]{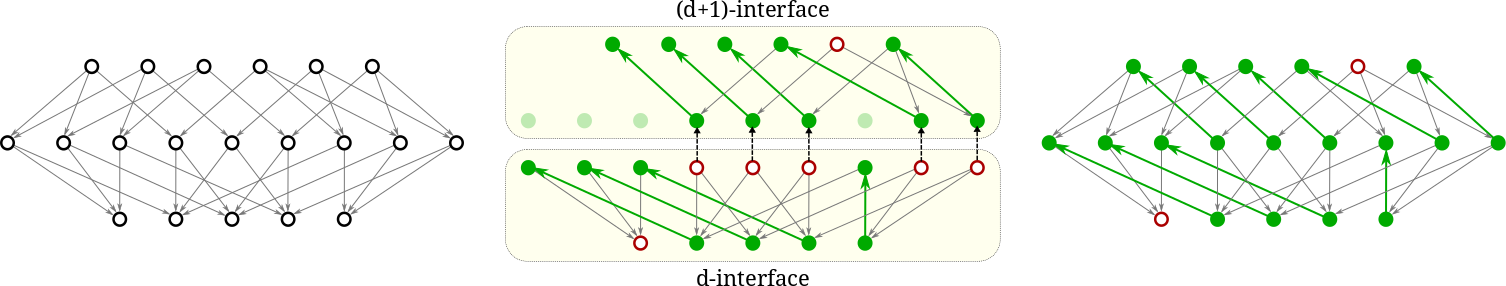}

\caption{An example illustrating Algorithm 2. Left: The input Hasse graph.
Middle: The $d$\selectlanguage{american}%
-\selectlanguage{english}%
simplices matched during execution of the algorithm on the $d$\foreignlanguage{american}{-interface}
are deleted from the ($d+1$\foreignlanguage{american}{)}\selectlanguage{american}%
-\selectlanguage{english}%
interface. Right: The final matching is obtained by combining matchings
from all interfaces.}

\label{fig:interface}
\end{figure}

\selectlanguage{american}%
Algorithm~\ref{alg:interface} , exploits special structures at the
lowest and highest interface. For instance, for any $D$-dimensional
manifold, there are well known algorithms in literature~\cite{BLPS13,Le03b,JP06}
for designing optimal gradient vector field for the $1$-interface
and the $D$-interface. See Appendix A of ~\cite{BLPS13}. As noted
in~\cite{BLPS13}, we can associate a special graph structure to the
$D$-interface. 

\begin{defn}[Dual Graph] The dual graph of a simplicial $D$-dimensional manifold $\mathcal{K}$, denoted by $\Gamma(\mathcal{K})$ is the graph whose vertices represent the $D$-simplices of $\mathcal{K}$ and whose edges join two $D$-simplices with a common $(D-1)$-facet.
\end{defn}

For the sake of completeness, we describe the optimal algorithms for
$1$-interface and the $D$-interface in Appendix~\ref{sec:Optimal-Algorithms-for}. 

Like in Algorithm~\ref{alg:frontier}, we first obtain the Hasse graph
of complex $\mathcal{K}$. We extract the $d$-interface of the Hasse
graph. For the $d$-interface, we design a Morse matching and reorient
the output graph $\mathcal{H_{V}}$ based on it. We then delete all
the regular simplices of the $d$-interface and the critical $(d-1)$-simplices.
This updated Hasse graph is available for the next iteration when
the $(d+1)$-interface is extracted and so on. For the $1$-interface
and the $D$-interface, the optimal algorithms are applied to design
the Morse matchings. For $d$-interfaces where $1<d<D$, procedure
\textbf{intermediateApx}() of Algorithm~\ref{alg:interface} is applied
to design the Morse matchings. 

We now describe procedure \textbf{intermediateApx}()\textbf{ }for
designing gradient vector field on the $d$-interface \textbf{$\mathcal{G}^{d}$.}
Algorithm 1 is essentially a maximum-matching followed by BFS-style
cycle removal and hence can be performed on any bipartite graph. In
particular, we apply it on graph $\mathcal{G}^{d}$ for $1<d<D$.
After cycle removal (from Algorithm 1) we may have a situation where
we have an unmatched simplex $\tau$ such that all its cofacets are
also unmatched. In that case, we match $\tau$ with one of its cofacets.
We perform this operation for all unmatched $(d-1)$-simplices whose
cofacets are also unmatched. This completes Morse matching for the
$d$-interface. In procedure \textbf{deleteAndReorient}(), if $\sigma^{d}$
is incident on simplex $\tau^{d-1}$ and if $\tau$ is regular at
the $(d-1)$-interface then we are justified in deleting it while
processing the $d$-interface, since $\tau$ is a regular sink node
for $d$-interface. The deletion of critical nodes does not affect
the behavior of Algorithm~\ref{alg:interface} per se. We delete them
here because the procedure \textbf{intermediateApx}() is used as a
subroutine in Algorithm~\ref{alg:minfacet} where this deletion is
crucial.

\begin{lem} \label{lem:interfaceacyclic} The orientation of $\mathcal{G}^{d}$ as computed by Algorithm 2 is acyclic.
\end{lem}
\begin{proof} Algorithm 1 provides an acyclic matching-based orientation of $d$-interface $\mathcal{G}^{d}$. So step 1 of \textbf{intermediateApx}() does not introduce any cycles. Now consider an unmatched simplex $\tau^{d-1}$ such that all its cofacets $\sigma_{1}^{d}\dots\sigma_{K}^{d}$ are also unmatched. For a directed acyclic graph there is an ordering relation $\alpha>\beta$ if there is a directed path from vertex $\alpha$ to vertex $\beta$. Clearly there is no ordering relation among $\sigma_{1}^{d}\dots\sigma_{K}^{d}$ since they are all critical. Introduction of the matching $\left\langle \tau,\sigma_{i}\right\rangle $ introduces the ordering relations of the type $\sigma_{j}>\sigma_{i}$ for all $j\in[1,K]$ and $j\neq i$. Therefore matching introduced by step 2 does not introduce any cycles and hence the orientation of $\mathcal{G}^{d}$ as computed by Algorithm~\ref{alg:interface} is acyclic.
\end{proof}

\begin{lem}\label{lem:interfaceoutputacylic} The orientation of the output graph $\mathcal{H_V}$ is acyclic.
\end{lem}
\begin{proof} From Lemma~\ref{lem:interfaceacyclic}, we conclude that the orientation for every $d$-interface $\ensuremath{\mathcal{G}^{d}}$ where $1<d<D$ is acyclic. Further, optimal acyclic matchings are computed for $1$-interface and $D$-interface respectively. Combining these two facts and along with the observation that every directed path is restricted to a unique $d$-interface, we conclude that the orientation of output graph\emph{ $\mathcal{H_{V}}$} is acyclic. 
\end{proof}

Now we introduce an idea that will help us prove approximation bounds
for Algorithm~\ref{alg:interface}. For the $d$-interface $\mathcal{G}_{d}$,
let $\tau^{d-1}$ be a critical simplex and let the set of cofacets
of $\tau$ that are regular be $\left\{ \beta_{1},\beta_{2},\dots,\beta_{K}\right\} $.
From line~\ref{lst:line:cantriplet} of procedure \textbf{intermediateApx}(),
we know that this set is non-empty. Let $\beta_{i}$ where $i\in[1\dots K]$
be a cofacet of $\tau$ with minimum index after performing a topological
sort on the $d$-interface\footnote{We do not actually perform topological sort on the $d$-interface, but need it for making an argument.}.
Now let $\alpha_{i}$ be such that $\alpha_{i}\prec\beta_{i}$ and
$\left\langle \alpha_{i},\beta_{i}\right\rangle $ is a gradient pair.
Then we can associate a \emph{canonical triplet} $\left\langle \left\langle \alpha_{i},\beta_{i}\right\rangle ,\tau\right\rangle $
to critical simplex $\tau^{d-1}$. Note that such a unique canonical
triplet is associated to every critical $(d-1)$-simplex. 

\begin{lem}\label{lem:interfaceratio} Algorithm 2 computes a $\nicefrac{2}{(d+2)}$-factor approximation to the Max Morse Matching restricted to the $d$-interface, $1<d<D$, of the Hasse graph of the $D$-dimensional manifold. 
\end{lem}
\begin{proof} Let $\left\langle \alpha_{i}^{d-1},\beta_{i}^{d}\right\rangle$ be a gradient pair. $\beta_{i}^{d}$ has $(d+1)$ facets of which at least one (namely $\alpha_i$) is regular. Therefore, the gradient pair $\left\langle \alpha_{i},\beta_{i}\right\rangle$ appears in at most $d$ canonical triplets. We group $\left\langle \alpha_{i},\beta_{i}\right\rangle$ with all the critical $(d-1)$-simplices that contain $\left\langle \alpha_{i}^{d-1},\beta_{i}^{d}\right\rangle$ in their canonical triplets. Each critical $(d-1)$-simplex appears in a unique group. Each group contains at least two regular simplices and at most $d$-critical simplices. So for every group we have the following ratio
$$\frac{\#\textnormal{matched simplices}}{\#\textnormal{total simplices}} \geq \frac{2}{(d+2)}$$
Hence, we obtain the approximation ratio of $\nicefrac{2}{(d+2)}$ for the $d$-interface. 
\end{proof}

The minimum of the ratio $\nicefrac{2}{(d+2)}$ over all $d$, $1<d<D$
is $\nicefrac{2}{(D+1)}$. The $1$-interface contributes to a single
critical simplex when the optimal algorithm is employed (See Appendix~\ref{sec:Optimal-Algorithms-for}). 

Finally, we consider the $D$-interface in the lemma below.

\begin{lem}\label{lem:d2} At least one $(D-1)$-simplex is matched at the conclusion of construction of Morse matching at the $(D-1)$-interface.
\end{lem}
\begin{proof} Consider an arbitrary $(D-1)$-simplex $\alpha_1$. Let $\partial \alpha_1$ denote the simplical boundary of an arbitrary $(D-1)$-simplex $\alpha_1$. Clearly, $\partial \alpha_1$ is a $(D-2)$-manifold without boundary. Therefore, using Morse inequalities there exists at least one $(D-2)$-simplex in $\partial \alpha$ that remains unmatched upon construction of Morse matching at the $(D-2)$-interface. When we look at boundaries of all such $(D-1)$-simplices $\alpha_i$, we may find several unmatched $(D-2)$-simplices upon conclusion of Morse matching at $(D-2)$-interface. Since one or more $(D-2)$-simplices remain unmatched at the start of the construction of Morse matching at the $(D-1)$-interface, at least one $(D-2)$-simplex can be matched to a $(D-1)$-simplex (without introducing cycles).    
\end{proof}

\begin{lem} After constructing Morse matching at the $D$-interface, the following ratio holds true: $$\frac{\#\textnormal{matched simplices}}{\#\textnormal{total simplices}} \geq \frac{4}{D+3}$$ \label{lem:manratio}
\end{lem}
\begin{proof} 
 Let $k$ denote the number of $(D-1)$-simplices that are matched at the conclusion of construction of Morse matching at the $(D-1)$-interface. Using Lemma~\ref{lem:d2}, we have $k\geq1$.  \\
Now consider the dual graph structure of the $D$-interface. The vertex degree of the dual graph is bounded by $D+1$.  So the total number of edges in the dual graph is smaller than $\nicefrac{(D+1)}{2}$. Applying the optimal algorithm in Appendix~\ref{sec:Optimal-Algorithms-for}, ensures that we have only one critical simplex in the dual graph. If $N$ is the number of vertices in the dual graph, the following ratio holds true:      
$$\frac{\#\textnormal{matched simplices}}{\#\textnormal{total simplices}} \geq \frac{2(N-1)}{(\frac{D+1}{2}+1)(N-1)+1-k} \geq \frac{2(N-1)}{(\frac{D+1}{2}+1)(N-1)} \geq\frac{4}{D+3}$$
\end{proof}

Note that $\nicefrac{4}{(D+3)}>\nicefrac{2}{(D+1)}$ for all $D\geq3$.
So the worst ratio over all $d$-interfaces, where $1\leq d\leq D$,
is $\frac{2}{(D+1)}$. Since the optimal number of regular simplices
is bounded by the total number of simplices, we get the following
theorem.

\begin{thm} For $D\geq3$, Algorithm 2 provides a $\nicefrac{2}{(D+1)}$-factor approximation for the Max Morse Matching problem for manifolds without boundary.
\end{thm}

We will like to make two remarks here regarding the approximation
factor. Firstly, the ratio is not affected by line~\ref{lst:line:canmatch}
(first step) of procedure \textbf{intermediateApx}(). It depends entirely
on line~\ref{lst:line:cantriplet} (second step) of \textbf{intermediateApx}().
We include a matching based preprocessing step prior to applying the
second step because in practice, doing so, gives significantly better
results. Secondly, the approximation ratio is over the \emph{total
number of simplices}. In that sense, Algorithm 2 and its analysis
helps further our understanding of combinatorial construction of manifolds.
In other words, irrespective of the complex size, the homology or
the presence of non-collapsible elements, we can always collapse
at least $\nicefrac{2}{(D+1)}$ number of simplices in that manifold!

\subsection{A $\nicefrac{2}{D}$-factor approximation algorithm for simplicial
manifolds}

\selectlanguage{english}%
\begin{figure}
\centering

\includegraphics[width=12cm]{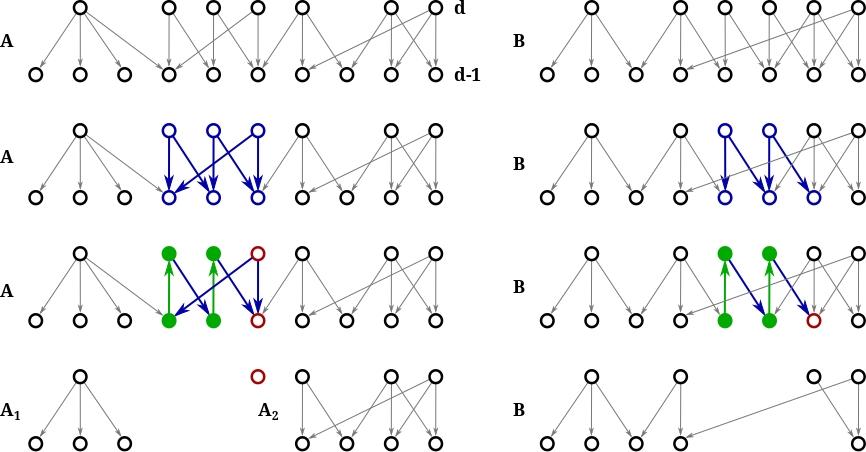}

\caption{Algorithm 3 processes the min\selectlanguage{american}%
-\selectlanguage{english}%
facet component in the $d$-interface (bold edges). Regular simplices
are denoted by filled vertices. Critical simplices and unprocessed
simplices are denoted by hollow vertices. Left: Deletion of $(d-1)$-simplices
of a min-facet component disconnects the graph. Right: Deletion of
$(d-1)$-simplices of the new min-facet component keeps the graph
connected. This process continues until none of the $d$-simplices
have any facets left.}

\label{fig:mincomp}
\end{figure}

\selectlanguage{american}%
Once again we restrict our attention to simplicial manifolds without
boundary. We build on Algorithm 2 by exploiting a finer substructure
within each interface to obtain a further improvement in ratio for
simplicial manifolds. We begin with some definitions.

\begin{defn}[facet degree, min-facet simplex of the $d$-interface] The number of facets incident on a simplex is defined as its \emph{facet degree}. For the $d$-interface, consider the subset of $d$-simplices $\mathcal{S}$ with at least one facet. We say that a $d$-simplex is a \emph{min-facet simplex} if over all simplices in $\mathcal{S}$, it has the minimum number of facets.  
\end{defn}

\begin{defn}[min-facet component of the $d$-interface] A \emph{min-facet component} is a subgraph of the $d$-interface that is a maximal connected graph induced by a set of min-facet simplices of the $d$-interface. 
\end{defn}

Like in Algorithm~\ref{alg:interface}, we process the Hasse graph
one $d$-interface at a time starting with the $1$-interface and
terminating with the $D$-interface. Also, like in Algorithm~\ref{alg:interface},
we use optimal algorithms to process the $1$-interface and the $D$-interface
of the Hasse graph. Only the intermediate interfaces are processed
differently. The procedure for handling intermediate interfaces is
outlined in Algorithm~\ref{alg:minfacet}.

\begin{algorithm}
\caption{The Min-Facet Component Algorithm}

\begin{algorithmic}[1]

\Procedure{interApxMinFacet}{$\mathcal{G}^d$}

\While{$\textbf{sizeOfMinFacet}(\mathcal{G}^d)>0$}
	\LState{$\mathcal{F_C}\leftarrow\textbf{extractMinFacetComponent}(\mathcal{G}^d)$}
	\LState{Apply \textbf{intermediateApx}($\mathcal{F_C},d$) from Algorithm~\ref{alg:interface}}
\EndWhile

\EndProcedure
\end{algorithmic}

\label{alg:minfacet} 
\end{algorithm}

By design, procedure \textbf{intermediateApx}() from Algorithm~\ref{alg:interface}
need not process the entire $d$-interface $\mathcal{G}^{d}$ at
one go. It may take any subgraph of the $d$-interface as its input.
The key idea is to iteratively compute Morse matching by executing
\textbf{intermediateApx}() on a min-facet component and after designing
a vector field on this component, we subsequently delete it from the
$d$-interface $\mathcal{G}^{d}$. As a consequence, $\mathcal{G}^{d}$
grows increasingly sparse and when the entire $d$-interface has
no edges left the while loop terminates. Figure~\ref{fig:mincomp}
illustrates sample executions of the Algorithm~\ref{alg:minfacet}. 

\begin{lem} \label{lem:grpath} If the $d$-interface of the Hasse graph is connected then there exists a gradient path connecting any two simplices $\alpha^{d-1}$ and $\beta^{d-1}$.
\end{lem}
\begin{proof} Since the $d$-interface is connected there exists a path in the $d$-interface graph connecting any two vertices $\alpha^{d-1}$ and $\beta^{d-1}$. However, it is easy to see that any path connecting $\alpha^{d-1}$ and $\beta^{d-1}$ in the $d$-interface graph is a gradient path in the discrete Morse theory sense.   
\end{proof}

\selectlanguage{english}%
\begin{figure}
\centering

\includegraphics[scale=0.5]{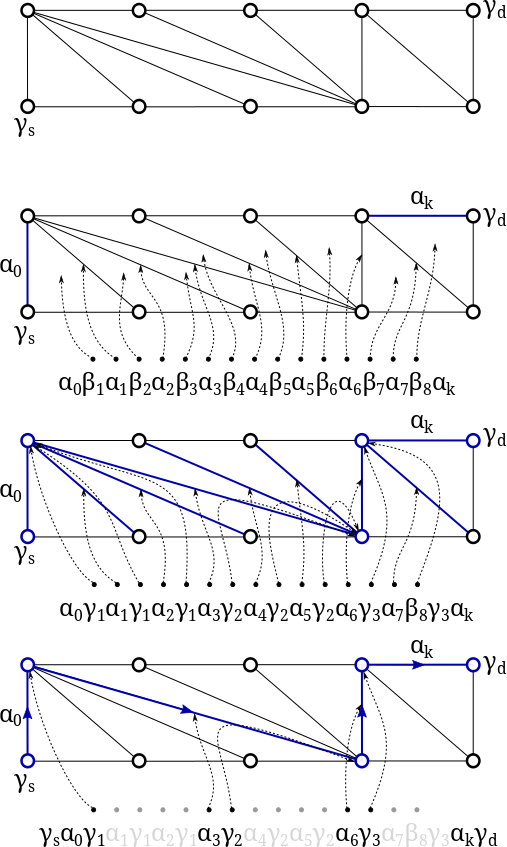}

\caption{In this figure, we wish to establish the connectivity of $\gamma_{s}$
and $\gamma_{d}$ in the $1$-interface. Let $\alpha_{0}$ and $\alpha_{k}$
be $1$-simplices containing $\gamma_{s}$ and $\gamma_{d}$ respectively.
It is known that the $2$-interface is connected. So, we can find
the gradient sequence $\alpha_{0}\beta_{1}\alpha_{1}\dots\beta_{8}\alpha_{k}$.
If we let $\gamma_{i}=\alpha_{i-1}\cap\alpha_{i}$, then we can extract
a new sequence $\alpha_{0}\gamma_{1}\alpha_{1}\dots\gamma_{8}\alpha_{k}$.
Finally, as explained in Lemma~\ref{lem:dminus}, this sequence can
be used to obtain subsequence $\gamma_{s}\alpha_{0}\gamma_{1}\alpha_{3}\gamma_{2}\alpha_{6}\gamma_{3}\alpha_{k}\gamma_{d}$
which establishes connectivity between $\gamma_{s}$ and $\gamma_{d}$.}

\label{fig:dminusminus}
\end{figure}
\begin{figure}
\centering

\includegraphics[width=12cm]{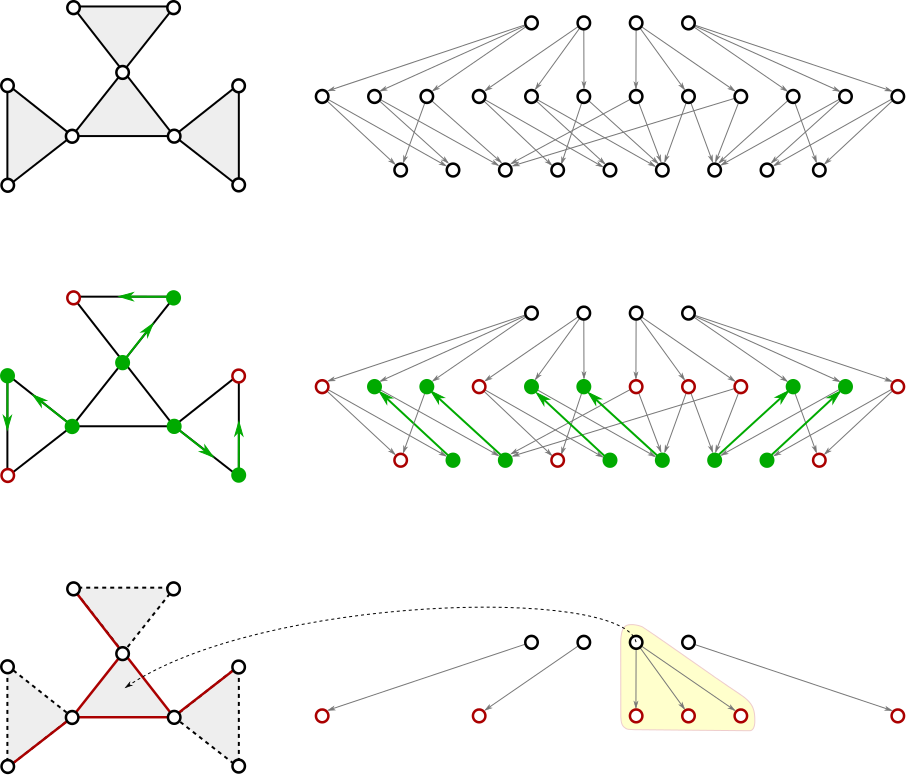}

\caption{Unlike in the case of manifolds without boundary, in this example
we have a complex whose $2$-interface is disconnected to begin with.
After designing vector field for the $1$-interface, suppose we delete
all the matched 1-simplices from the Hasse graph. Then there exists
a connected component in the $2$-interface for which all $2$-simplices
of that connected component has three facets (In this case, this connected
component comprises of a single simplex with all three solid edges). }

\label{fig:dplusminus}
\end{figure}

\selectlanguage{american}%
\begin{lem} \label{lem:dminus} If the $d$-interface of the Hasse graph is connected then the $(d-1)$-interface is connected. 
\end{lem}
\begin{proof} Suppose the $d$-interface of the Hasse graph is connected. If any two arbitrary $(d-2)$-simplices $\gamma_s$ and $\gamma_d$ can be shown to be connected, then the $(d-1)$-interface is connected.  To begin with let $\alpha_0$ be any $(d-1)$-simplex with $\gamma_s$ as its facet and $\alpha_k$ be any $(d-1)$-simplex with $\gamma_d$ as its facet. If $\alpha_0=\alpha_k$, there is nothing to prove. So for the remainder of the proof we shall assume that $\alpha_0\neq\alpha_k$. Since the $d$-interface is connected, using Lemma~\ref{lem:grpath}, there exists a gradient path  $\alpha_0,\beta_1,\alpha_1\dots,\alpha_{(k-1)},\beta_k,\alpha_k$ connecting $\alpha_0$ and $\alpha_k$ where all $\beta_i,i\in[1,k]$ are $d$-simplices and all $\alpha_i,i\in[0,k]$ are $(d-1)$-simplices. Now, since every $d$-simplex $\beta_i,i\in[1,k]$ is common to two $(d-1)$-simplices $\alpha_{(i-1)}$ and $\alpha_i$ belonging to the gradient path connecting $\alpha_0$ and $\alpha_k$, we know that $\alpha_{(i-1)}$ and $\alpha_i$ will share a facet which we denote by $\gamma_i^{(d-2)}$. In other words, we construct a new simplicial sequence $\mathcal{S}=\alpha_0,\gamma_1,\alpha_1\dots,\alpha_{(k-1)},\gamma_k,\alpha_k$ from gradient path $\alpha_0,\beta_1,\alpha_1\dots,\alpha_{(k-1)},\beta_k,\alpha_k$, where $\gamma_i=\alpha_{(i-1)}\cap\alpha_i$. However, note that in this case $\gamma_i$ may possibly be equal to $\gamma_j$ for some $i\neq j$. See Fig.\ref{fig:dminusminus} for an example. Without loss of generality assume $\gamma_s\neq\gamma_1$ and $\gamma_d\neq\gamma_k$. We prove connectivity of $\gamma_s$ and $\gamma_d$ by induction. For base case, we note that $\gamma_s$ is connected to $\gamma_1$ since $\gamma_1$ and $\gamma_s$ are facets of simplex $\alpha_0$. For induction step, suppose $\gamma_s$ is connected to $\gamma_i$. Now consider the next two elements in sequence $\mathcal{S}$ namely $\alpha_i$ and $\gamma_{(i+1)}$. If $\gamma_i=\gamma_{(i+1)}$, then $\alpha_i$ makes no contribution towards finding a path connecting $\gamma_s$ and $\gamma_d$ and hence we ignore it. Else if $\gamma_i\neq\gamma_{(i+1)}$, then both $\gamma_i$ and $\gamma_{(i+1)}$ are facets of $\alpha_i$ and hence $\gamma_i$ is connected to $\gamma_{(i+1)}$ in the $(d-1)$ interface. By transitivity, $\gamma_s$ is connected to $\gamma_{(i+1)}$, which completes the induction step. Finally, both $\gamma_k$ and $\gamma_d$ are facets of $\alpha_k$ and hence $\gamma_k$ is connected to $\gamma_d$. By transitivity $\gamma_s$ is connected to $\gamma_d$. This proves that there exists a subgraph of the $(d-1)$-interface that connects any two arbitrary $(d-2)$-simplices  $\gamma_s$ and $\gamma_d$. Hence proved.
\end{proof}

\begin{lem}\label{lem:manconnected} For a connected $D$-manifold without boundary, all $d$-interfaces are connected.
\end{lem}
\begin{proof} Let $K$ be the number of connected components of the $D$-interface. Then clearly $\beta_D=K$. Since a connected manifold without boundary has $\beta_D=1$, we conclude that the $D$-interface is connected. Combining this fact with Lemma~\ref{lem:dminus} implies that all $d$-interfaces are connected.
\end{proof}

\begin{lem}\label{lem:firstsmall} Following the design of gradient vector field for the $(d-1)$-interface, if the deletion of regular sinks of $d$-interface graph disconnects the $d$-interface, then every connected component has at least one simplex with facet degree smaller than $d+1$.
\end{lem}
\begin{proof} From Lemma~\ref{lem:manconnected} we know that the $d$-interface is a single connected component to begin with. Suppose that the regular sinks of $d$-interface graph are deleted in some sequence. Suppose that $\gamma^{d-1}$ is the first simplex whose deletion disconnects the $d$-interface graph. Then every connected component (in the traditional graph theory sense) has at least one $d$-simplex which is incident on $\gamma^{d-1}$ and hence upon deletion of $\gamma^{d-1}$, every connected component will have at least one simplex with  facet degree smaller than $(d+1)$. The same argument can be continued for subsequent deletions and resulting disconnections.     
\end{proof}

To see that Lemma~\ref{lem:manconnected} is essential for Lemma~\ref{lem:firstsmall}
to work, we see an example in Fig.~\ref{fig:dplusminus} where lack
of connectivity in the $d$-interface leads to components (in the
$d$-interface) with minimum facet degree equal to $(d+1)$. 

\begin{lem}\label{lem:alwayssmall} For a $d$-interface, every min-facet component has facet degree bounded by $d$.
\end{lem}
\begin{proof} We prove the claim by induction. \\
Base Case: From Lemma~\ref{lem:firstsmall}, having deleted all the regular sinks of the $d$-interface, there exists at least one simplex with facet degree bounded by $d$ in every connected component. We arbitarily choose a min-facet simplex in one of the connected components of the $d$-interface and discover the min-facet component around it by exploring neighboring $d$-simplices iteratively. Such a min-facet component has facet degree bounded by $d$. We design vector field on this min-facet component and subsequently delete it from the $d$-interface.\\
Induction Hypothesis: Suppose that we have designed a vector field on $(i-1)$ min-facet components and subsequently deleted them. Each time there exists at least one simplex with facet degree bounded by $d$ in every connected component. 
Induction Step: Now we discover the $i^{th}$ min-facet component say $\mathcal{F}_i$. Suppose this min-facet component belongs to some connected component $\mathcal{C}_j$. \\
Case 1: If $\mathcal{F}_i$ consists of all vertices in $\mathcal{C}_j$ then after processing it and deleting its vertices the other connected components continue to satisfy the facet degree (bounded by $d$) condition. So there is nothing to prove. \\
Case 2: If $\mathcal{F}_i$ consists of all the $(d-1)$-simplices of $\mathcal{C}_j$. Then upon deletion of $\mathcal{F}_i$, all $d$-simplices of $\mathcal{C}_j\setminus\mathcal{F}_i$ will have zero facet degree and we will attempt to match the $d$-simplices of $\mathcal{C}_j\setminus\mathcal{F}_i$ for the $(d+1)$-interface.\\
Case 3: Suppose if $\mathcal{F}_i\subsetneq \mathcal{C}_j$ and $\mathcal{C}_j\setminus\mathcal{F}_i$ has one or more $(d-1)$-simplices. Clearly there exists at least one $d$-simplex say $\sigma$ in $\mathcal{C}_j\setminus\mathcal{F}_i$ with at least one edge incident on a $(d-1)$-simplex in $\mathcal{F}_i$ and at least one edge incident on  a $(d-1)$-simplex in $\mathcal{C}_j\setminus\mathcal{F}_i$. Having designed a gradient vector field on $\mathcal{F}_i$, we delete the regular simplices and the critical $(d-1)$ simplices belonging to $\mathcal{F}_i$. Now we consider two subcases that are illustrated in Fig.~\ref{fig:mincomp} \\
Case 3a: Consider the case when $\mathcal{C}_j$ stays connected after deleting the $i^{th}$ min-facet component. In this case the facet degree of $\sigma$ will reduce by at least 1 and hence the facet degree of $\sigma$ is bounded by $d$. There may be other simplices in $\mathcal{F}_i\subsetneq \mathcal{C}_j$ whose facet degree may also reduce. All other connected components are unaffected. So, every component will have min-facet degree bounded by $d$. \\
Case 3b: Now consider the case where upon deletion of $\mathcal{F}_i$, $\mathcal{C}_j$ splits into several components. Imagine that we are not deleting the simplices of $\mathcal{F}_i$ all at once, but sequentially. Making an argument along the lines of Lemma~\ref{lem:firstsmall}, we conclude that irrespective of which connected component the min-facet component is chosen from, it will have facet degree bounded by $d$.             
\end{proof} 

\begin{lem}\label{lem:minacyclic} An orientation of the min-facet component $\mathcal{F_{C}}$ based on the matchings computed by procedure interApxMinFacet() is acylic. 
\end{lem}

The proof of the Lemma~\ref{lem:minacyclic} is identical to the proof
of Lemma~\ref{lem:interfaceacyclic}.

\begin{lem}\label{lem:allacyclic} An orientation of the $d$-interface of output graph $\mathcal{H_V}$ based on the matchings computed by procedure interApxMinFacet() is acylic. 
\end{lem}
\begin{proof} We prove this claim by induction. We use a condition namely the \emph{vertex deletion criterion} which says that: For the $d$-interface, a $(d-1)$-simplex satisfies the vertex deletion condition if and only if all paths that go through that simplex end up in a sink.\\
\textbf{Base Case}: Suppose that we are processing the first min-facet component for the $d$-interface. From Lemma~\ref{lem:minacyclic}, we know that an orientation of edges of a min-facet component is acyclic. For this orientation, a path from any vertex in the component ends up in a sink. Therefore if we were to delete all the $(d-1)$-simplices in the min-facet component, we obey the vertex deletion criterion. If graph $\mathcal{H_V}$ is oriented based on the matchings found in the first min-facet component, then it is acyclic.\\
\textbf{Induction Step}: Suppose that we have processed $i$ min-facet components and suppose that we have used these min-facet components to orient the $d$-interface of $\mathcal{H_V}$ and so far it is found to be acyclic. Also, the vertices deleted so far are those that have satisfied the vertex deletion condition. Now suppose we have extracted the $(i+1)^{th}$ min-facet component say $\mathcal{F}_{i+1}$. While the edges that lead to sinks maybe absent in min-facet component, $\mathcal{F}_{i+1}$, the corresponding $d$-simplices in output graph $\mathcal{H_V}$ will have these edges. If we restrict our attention to undeleted edges, then from Lemma~\ref{lem:minacyclic}, the orientation of edges of $(i+1)^{th}$ min-facet component itself is acyclic i.e. all paths will lead strictly to critical sinks of $\mathcal{F}_{i+1}$. But if we look at the corresponding orientation in $\mathcal{H_V}$, the paths emanating from a $(d-1)$ simplex of $\mathcal{F}_{i+1}$ will either end up in critical sinks of $\mathcal{F}_{i+1}$ (through undeleted edges) or in regular/critical sinks of $\mathcal{F}_{j}$ for $j<(i+1)$ (through deleted edges). In any case, all paths going from $(d-1)$-simplices of $\mathcal{F}_{i+1}$ go to sinks thereby satisfying the vertex deletion criterion. Also designing gradient field on $\mathcal{F}_{j}$ does not introduce cycles in $\mathcal{H_V}$. Morse matching on the $d$-interface is designed when all the min-facet components are processed and deleted. Since none of them introduce cycles, we say that output graph $\mathcal{H_V}$ is acyclic.  
\end{proof}

\begin{lem} \label{lem:mainratio} For the $d$-interface, the ratio $\frac{\#\textnormal{matched simplices}}{\#\textnormal{total simplices}}\geq\frac{2}{d+1}$.
\end{lem}
\begin{proof} 
The proof is identical to that of Lemma~\ref{lem:interfaceratio} except for one important difference. In case of Algorithm 2, for a $d$-interface every $d$-simplex has $d+1$ facets. But according to Lemma~\ref{lem:alwayssmall}, for a min-facet component the facet degree is bounded by $d$. Using the notion of canonical triplets for a min-facet component, for every gradient pair we get at most $(d-1)$ critical simplices. So the ratio $\frac{\#\textnormal{matched simplices}}{\#\textnormal{total simplices}}\geq\frac{2}{d+1}$ for every min-facet component. 
For a $d$-interface, every $(d-1)$-simplex is part of some min-facet component and is classifed as a regular simplex or as a critical simplex and subsequently deleted from the $d$-interface. Therefore the bound of $\nicefrac{2}{(d+1)}$ carries over from min-facet components to $d$-interfaces.   
\end{proof}

If we take the minimum for the ratio $\nicefrac{2}{d+1}$ over all
$d$ such that $1<d<D$, we get $\nicefrac{2}{D}$. By Lemma~\ref{lem:manratio},
for a $D$-interface the ratio $\ensuremath{\frac{\textnormal{\#matched simplices}}{\#\textnormal{total simplices}}}$
is equal to $\frac{4}{D+3}$. Note that $\nicefrac{4}{(D+3)}\geq\nicefrac{2}{D}$
for all $D\geq3$. So the worst ratio of $\ensuremath{\frac{\#\textnormal{matched simplices}}{\#\textnormal{total simplices}}}$
over all $d$-interfaces where $1\leq d\leq D$ is $\frac{2}{D}$.
Since the optimal number of regular simplices $\leq$ total number
of simplices, we get the following theorem.

\begin{thm} For $D\geq3$, Algorithm~\ref{alg:minfacet} provides a $\nicefrac{2}{D}$-factor approximation for the Max Morse Matching problem.
\end{thm}

\subsubsection{Approximation bound for nonmanifold complexes\label{sub:apx-non-man-comp}}

\selectlanguage{english}%
Note that \foreignlanguage{british}{Algorithm~\ref{alg:minfacet}}
can be applied to non\selectlanguage{british}%
-\selectlanguage{english}%
manifold complexes as well \foreignlanguage{british}{if we apply the
optimal algorithm for the $1$-interface and procedure \textbf{intermediateApx}()
for the remaining interfaces. To prove a bound for non\selectlanguage{english}%
-\selectlanguage{british}%
manifold complexes, we need to do a slightly different kind of analysis.
We begin with a few definitions. Let $T$ denote the set of all $(D-1)$-simplices
of the Hasse graph. Let $B$ denote the $(D-1)$-simplices that have
been paired with $(D-2)$-simplices and let $|A|=|T|-|B|$. Let $\mathcal{R}_{D}$
denote the set of regular simplices found by Algorithm~\ref{alg:minfacet}
at the $D$-interface. We now establish a relation between $|\mathcal{R}_{D}|$
and $|A|$.}

\selectlanguage{british}%
\begin{lem}\label{lem:dcon} $|\mathcal{R}_{D}|\geq \frac{2}{r} |A|$ where $r=D$ if the $D$-interface is connected and $r=(D+1)$ if the $D$-interface is not connected.\footnote{Note that the $D$-interface of a general non-manifold simplicial complex may or may not be connected.} 
\end{lem}
\begin{proof} First we consider the case when the $D$-interface is not connected at the start of Algorithm~\ref{alg:minfacet}. At each stage of Algorithm~\ref{alg:minfacet}, the minimum facet-degree of a simplex is not more than $(D+1)$. Once again we use the idea of canonical triplets. Every critical $(D-1)$-simplex occurs in a unique canonical triplet. Also, every regular $(D-1)$-simplex occurs in at most $D$ canonical triplets. So every regular $(D-1)$-simplex corresponds to a set of at most $D$ critical $(D-1)$-simplices. Together they make up the entire set $A$. Hence we have $|\mathcal{R}_{D}|\geq \frac{2}{D+1} |A|$. \\
Now suppose the $D$-interface is connected at the start of the Algorithm. Then Lemma~\ref{lem:firstsmall} and Lemma~\ref{lem:alwayssmall} apply and the minimum facet-degree of a simplex (in a min-facet component) is not more than $D$. In this case, a regular $(D-1)$-simplex occurs in at most $(D-1)$ canonical triplets. Accordingly, $|\mathcal{R}_{D}|\geq \frac{2}{D} |A|$.
\end{proof}

Now, let $\mathcal{R}$ denote the set of all regular simplices found
by Algorithm~\ref{alg:minfacet} and and let $\mathcal{R}_{L}=\mathcal{R}-\mathcal{R}_{D}$.

Let $\mathcal{S}_{D-2}$ denote the set of vertices of the Hasse graph
that belong to one of the $d$-interfaces where $1\leq d\leq(D-2)$
and $\mathcal{S}_{D-1}$ denote the set of vertices of the Hasse graph
that belong to one of the $d$-interfaces where $1\leq d\leq(D-1)$.
Let $\mathcal{S=}\mathcal{S}_{D-2}\cup B$. Finally, let $|\mathcal{S}|$
denote the cardinality of vertex set $\mathcal{S}$ and $|\mathcal{S}_{D-1}|$
denote the cardinality of vertex set $\mathcal{S}_{D-1}$.

\begin{lem} \label{lem:obvious} $|\mathcal{R}_{L}| \geq \frac{2}{D}|\mathcal{S}|$
\end{lem}
\begin{proof} Let $\mathcal{G_S}$ be the graph induced by set $\mathcal{S}$. Note that every simplex belonging to graph $\mathcal{G_S}$ occurs in some canonical triplet. In particular this happens to be true since all $(D-1)$-simplices of $\mathcal{S}$ are matched by Algorithm~\ref{alg:minfacet}.  Using Lemma~\ref{lem:mainratio} for Algorithm~\ref{alg:minfacet} applied to $\mathcal{G_S}$, the ratio $\frac{\#\textnormal{matched simplices}}{\#\textnormal{total simplices}}\geq\frac{2}{D}$. In other words, we get,  $|\mathcal{R}_{L}| \geq \frac{2}{D}|\mathcal{S}|$.
\end{proof}

Let $\mathcal{O}$ denote the cardinality of regular nodes found by
optimal Morse Matching. 

\begin{lem}\label{lem:optlast} $\mathcal{O}\leq |\mathcal{S}_{D-1}|+|T|$
\end{lem}
\begin{proof} The maximum number of simplices of the $D$-level that can be matched by any algorithm is bounded by |T| i.e. the total number of simplices of the $(D-1)$-level. Also the set $\mathcal{S}_{D-1}$ consists of all simplices of the Hasse graph except those that belong to the $D$-level. So, the optimal algorithm can not possibly match more than $|\mathcal{S}_{D-1}|+|T|$ number of simplices of the Hasse graph.
\end{proof}

We now consider the expression $\ensuremath{D|\mathcal{R}_{L}|+r|\mathcal{R}_{D}|}$,

\begin{eqnarray*} 
D|\mathcal{R}_{L}|+r|\mathcal{R}_{D}| & \geq & 2|S|+2|A| \dots\dots\dots\dots\dots\dots\textnormal{using Lemma~\ref{lem:dcon} and Lemma~\ref{lem:obvious}}\\
  									& \geq & |S|+ |B| + 2|A| \dots\dots\dots\dots\textnormal{using the fact that }B \subseteq S\\
									  & \geq & (|S|+|A|) + (|B|+|A|)\\
									  &   =  & |\mathcal{S}_{D-1}|+|T|\dots\dots\dots\dots\dots\textnormal{by definition}\\
									  & \geq & O \dots\dots\dots\dots\dots\dots\dots\dots\dots\textnormal{using Lemma~\ref{lem:optlast}}
\end{eqnarray*} 

If the $D$-interface is connected we get, 

$\ensuremath{D|\mathcal{R}}|=\ensuremath{D|\mathcal{R}_{L}|+D|\mathcal{R}_{D}|}\geq\mathcal{O}$
i.e. $\ensuremath{|\mathcal{R}}|\geq\frac{1}{D}\mathcal{O}$

$\vphantom{}$

If the $D$-interface is not connected we get, 

$(D+1)\mathcal{R}=(\ensuremath{D+1)|\mathcal{R}_{L}|+(D+1)|\mathcal{R}_{D}|}\geq D|\mathcal{R}_{L}|+(D+1)|\mathcal{R}_{D}|\geq\mathcal{O}$
i.e. $\ensuremath{|\mathcal{R}}|\geq\frac{1}{(D+1)}\mathcal{O}$.

$\vphantom{}$

Therefore, for non-manifold complexes, Algorithm~\ref{alg:minfacet}
gives a $\nicefrac{1}{D}$ approximation if the $D$-interface is
connected and a $\nicefrac{1}{(D+1)}$ approximation if the $D$-interface
is not connected.

Likewise one can obtain $\nicefrac{1}{(D+1)}$ approximation bound
for Algorithm~\ref{alg:interface}, irrespective of whether or not
the complex is connected. 

\selectlanguage{american}%

\section{Experimental results \label{sec:experiments}}

\selectlanguage{english}%
All three approximation algorithms proposed in this paper are implemented
in Java. In this section, we describe results of experiments comparing
the algorithms proposed here with three algorithms for Morse matching,
namely reduction heuristic, coreduction heuristic a naïve approximation
algorithm that provides an approximation ratio of $1/(D+1)$. A prototype
implementation was used to observe the practical performance of these
algorithms on more than 800 complexes. We used both synthetic random
datasets and complexes generated by Hachimori~\cite{Hach01} (also
used in earlier work~\cite{BLPS13}) and Lutz~\cite{Lutzlib}, for
experiments. Random complexes were generated according to the method
described by Meshulam and Wallach~\cite{Mesh09} and a variant. In
the variant, we select a random number of valid $d$\selectlanguage{british}%
-\selectlanguage{english}%
simplices for all $1\leq d\leq D$ instead of selecting a random number
of $D$\selectlanguage{british}%
-\selectlanguage{english}%
simplices and inferring all faces. We refer to the the complexes generated
by this variant as Type~2 random complexes. For additional details
on these complexes see Section~\ref{sub:Type-2-random}. 

It is clear that the quantity 2$|\mathcal{M}|$ where $|\mathcal{M}|$
is the size of maximum cardinality matching as well as the quantity
$N-\Sigma\beta_{i}$ which is equal to the difference between number
of simplices and the sum of Betti numbers, provide conservative upper
bounds on the number of regular cells in the optimal Morse matching.
Let $\mathcal{R}$ be the set of regular simplices generated by a
Max Morse approximation algorithm. We \emph{estimate} the quality
of the approximation using the ratio $\frac{|\mathcal{R}|}{Min(2|\mathcal{M}|,N-\Sigma\beta_{i})}$.
Tables~\ref{tab:hachimoriLib},~\ref{tab:selManifold},~\ref{tab:randomComplexes}~and~\ref{tab:type2complexes}
list estimated approximation ratios on selected datasets. Algorithm~3
consistently provided the best ratios, always greater than 0.93 for
all 300 random complexes in our dataset. For more than 450 manifolds
from the Lutz dataset, Algorithm~3 reported worst estimated approximation
ratio of $0.969$. Algorithm~3 provided optimal estimated approximation
ratio for 56\% of manifolds from the Lutz dataset. These results suggest
that Algorithm~3 not only provides good theoretical bounds, but also
performs well practically. 

In sections that follow, we first discuss a na\"ive approximation
algorithm followed by experiments on datasets from four different
sources. 

\selectlanguage{american}%

\subsection{A $\nicefrac{1}{(D+1)}$-factor \foreignlanguage{english}{Na\"ive}
Approximation Algorithm}

Consider the following approximation algorithm: Given a simplicial
complex $\mathcal{K}$ compute its Hasse graph $\mathcal{H_{K}}$.
Perform cardinality matching on graph $\mathcal{H_{K}}$ and obtain
the matching based reorientation $\overline{\HCC_{\KCC}}$ . Include
all the down-edges of $\overline{\HCC_{\KCC}}$ in the output graph
$\mathcal{H_{O}}$.
\begin{enumerate}
\item Pick an arbitrary up-edge $e$ and include it in $\mathcal{H_{O}}$. 
\item Include the reversed orientations of all the leading up-edges of $e$
in $\mathcal{H_{O}}$.
\item Remove up-edge $e$ and the leading up-edges of $e$ from $\overline{\HCC_{\KCC}}$ 
\item Repeat steps 1-3 until all up-edges of $\overline{\HCC_{\KCC}}$ are
exhausted. 
\end{enumerate}
Clearly, $\mathcal{H_{O}}$ has no cycles because none of the up-edges
in $\mathcal{H_{O}}$ has leading up-edges. Also, for every up-edge
that we select, we reverse at most $D$ up-edges. Since cardinality
matching is an upper bound on optimal value of Max Morse Matching,
we get an approximation ratio of $\nicefrac{1}{(D+1)}$ for this algorithm.

At the outset, the ratio $\nicefrac{(D+1)}{(D^{2}+D+1)}$ obtained
by Algorithm~\ref{alg:frontier} does not seem to be a significant
improvement over $\nicefrac{1}{(D+1)}$. However, as we shall witness
in sections that follow, the \foreignlanguage{english}{estimated approximation
ratios observed for the na\"ive algorithm are significantly worse
in practice. In fact, in order to ensure that the approximation algorithms
designed for Max Morse Matching problem remains relevant for applications
like homology computation, scalar field topology etc., we need to
design algorithms that can be shown to have good theoretical approximation
and complexity bounds combined with competitive estimated approximation
ratios.}

\subsection{Coreduction and Reduction heuristics}

The coreduction heuristic for constructing Morse matchings was introduced
in \foreignlanguage{british}{~\cite{HMMNWJD10}.} In this section,
we briefly describe reduction and coreduction heuristics for constructing
Morse matchings on simplicial complexes for the sake of completeness.\foreignlanguage{british}{
Suppose we are given a simplicial complex $\mathcal{K}$. We first
describe the coreduction heuristic.}

Perfom the folowing steps until complex $\mathcal{K}$ is empty:
\begin{enumerate}
\item if there is a simplex $\alpha$ with a free coface $\beta$ available,
include the pair $\langle\alpha,\beta\rangle$ in the set of Morse
matchings and delete $\alpha$ and $\beta$ from $\mathcal{K}$.
\item if no such simplex with a free coface is available then we select
a simplex $\gamma$ of lowest available dimension and make it critical.
We then delete $\gamma$ from $\mathcal{K}$.
\end{enumerate}
We now describe the reduction heuristic.

Perfom the folowing steps until complex $\mathcal{K}$ is empty:
\begin{enumerate}
\item if there is a simplex $\beta$ with a free face $\alpha$ available,
include the pair $\langle\alpha,\beta\rangle$ in the set of Morse
matchings and delete $\alpha$ and $\beta$ from $\mathcal{K}$.
\item if no such simplex with a free coface is available then we select
a simplex $\gamma$ of highest available dimension and make it critical.
We then delete $\gamma$ from $\mathcal{K}$.
\end{enumerate}
\selectlanguage{english}%

\subsection{Experiments on the Hachimori Dataset}

This dataset consists of complexes downloaded from Hachimori's collection
of simplicial complexes\footnote{\href{http://infoshako.sk.tsukuba.ac.jp/~hachi/math/library/index_eng.html}{http://infoshako.sk.tsukuba.ac.jp/$\sim$hachi/math/library/index\_{}eng.html}}.
Table~\ref{tab:hachimoriLib} lists the observed approximation ratios
for all the algorithms. For complexes in Table~\ref{tab:hachimoriLib},
maximum size of $\Sigma\beta_{i}$ is 2. Clearly, coreduction heuristic
provided the best approximation ratios for this dataset. However,
Algorithm~3 reported ratios comparable to coreduction. Algorithm~3
reports optimal Morse matching for 7 of the 20 complexes in this dataset,
while coreduction gives optimal result for 10 complexes.

\begin{table}[th]
\centering

\begin{tabular}{cccccccc}
\toprule 
\multirow{2}{*}{Input} & \multirow{2}{*}{N} & \multicolumn{6}{c}{Estimated approximation ratios}\tabularnewline
\cmidrule{3-8} 
 &  & Na\"ive  & Algo 1 & Algo 2 & Algo 3 & Cored & Red\tabularnewline
\midrule
\midrule 
\multicolumn{8}{l}{\textbf{2D complexes}}\tabularnewline
\midrule
projective & 31 & 0.800 & \textbf{0.933} & \textbf{0.933} & \textbf{0.933} & \textbf{0.933} & \textbf{0.933}\tabularnewline
dunce\_hat & 49 & 0.667 & 0.917 & \textbf{0.958} & \textbf{0.958} & \textbf{0.958} & 0.917\tabularnewline
bjorner & 32 & 0.667 & 0.933 & \textbf{1.000} & \textbf{1.000} & \textbf{1.000} & 0.867\tabularnewline
nonextend & 39 & 0.632 & 0.895 & 0.947 & \textbf{1.000} & \textbf{1.000} & 0.895\tabularnewline
c-ns & 75 & 0.703 & 0.892 & \textbf{0.946} & \textbf{0.946} & \textbf{0.946} & 0.865\tabularnewline
c-ns2 & 79 & 0.615 & 0.897 & 0.974 & 0.974 & \textbf{1.000} & 0.846\tabularnewline
c-ns3 & 63 & 0.667 & 0.871 & \textbf{0.968} & \textbf{0.968} & \textbf{0.968} & 0.903\tabularnewline
simon & 41 & 0.750 & 0.950 & \textbf{0.950} & \textbf{0.950} & \textbf{0.950} & 0.850\tabularnewline
simon2 & 31 & 0.667 & 0.800 & \textbf{0.933} & \textbf{0.933} & \textbf{0.933} & 0.867\tabularnewline
\midrule 
\multicolumn{8}{l}{\textbf{3D complexes}}\tabularnewline
\midrule
poincare & 392 & 0.651 & 0.933 & 0.954 & 0.979 & \textbf{0.990} & 0.923\tabularnewline
knot & 6,203 & 0.628 & 0.942 & 0.940 & 0.997 & \textbf{1.000} & 0.927\tabularnewline
bing & 8,131 & 0.640 & 0.946 & 0.943 & 0.997 & \textbf{0.999} & 0.933\tabularnewline
nc\_sphere & 8,474 & 0.616 & 0.941 & 0.945 & 0.989 & \textbf{1.000} & 0.937\tabularnewline
rudin & 215 & 0.617 & 0.935 & 0.944 & \textbf{1.000} & \textbf{1.000} & 0.925\tabularnewline
gruenbaum & 167 & 0.663 & 0.928 & 0.928 & \textbf{1.000} & \textbf{1.000} & 0.904\tabularnewline
ziegler & 119 & 0.695 & 0.983 & 0.915 & \textbf{1.000} & \textbf{1.000} & 0.864\tabularnewline
lockeberg & 216 & 0.636 & 0.944 & 0.972 & \textbf{1.000} & \textbf{1.000} & 0.897\tabularnewline
mani-walkup-C & 464 & 0.645 & 0.944 & 0.922 & \textbf{1.000} & \textbf{1.000} & 0.922\tabularnewline
mani-walkup-D & 392 & 0.621 & 0.923 & 0.923 & \textbf{0.990} & \textbf{0.990} & 0.908\tabularnewline
\midrule 
\multicolumn{8}{l}{\textbf{5D complexes}}\tabularnewline
\midrule
nonpl\_sphere & 2,680 & 0.554 & 0.841 & 0.883 & 0.989 & \textbf{0.997} & 0.954\tabularnewline
\bottomrule
\end{tabular}

\caption{Observed approximation ratios for Hachimori's Simplicial Complex Library.
\foreignlanguage{american}{N indicates the number of simplices in
the complex. Cored refers to Coreduction Algorithm, Red refers to
Reduction Algorithm. }For a given input, the best estimated approximation
ratios across all algorithms tested are highlighted in bold.}

\label{tab:hachimoriLib}
\end{table}

\subsection{Experiments on the Lutz Manifold Dataset}

The second dataset consists of manifolds of dimensions ranging from
3 to 11. These manifolds were downloaded from a library of manifolds
created by Lutz\footnote{The Manifold page: \href{http://page.math.tu-berlin.de/~lutz/stellar/vertex-transitive-triangulations.html}{http://page.math.tu-berlin.de/$\sim$lutz/stellar/vertex-transitive-triangulations.html}}. 

Table~\ref{tab:selManifold} lists approximation ratios observed
for selected complexes within this dataset. For manifolds in Table~\ref{tab:selManifold},
maximum size of $\Sigma\beta_{i}$ is 14 whereas the average size
of $\Sigma\beta_{i}$ is 5.07. Coreduction heuristic provided the
best approximation ratios for this dataset. However, Algorithm~3
matched the performance of coreduction heuristic for many complexes
and in some cases outperformed coreduction. Also, Algorithm~3 was
consistently better than reduction heusritic.

\begin{table}[th]
\centering

\begin{tabular}{cccccccc}
\toprule 
\multirow{2}{*}{Input} & \multirow{2}{*}{$N$} & \multicolumn{6}{c}{Estimated approximation ratios}\tabularnewline
\cmidrule{3-8} 
 &  & Na\"ive  & Algo 1 & Algo 2 & Algo 3 & Cored & Red\tabularnewline
\midrule
\midrule 
3\_12\_13\_3 & 192 & 0.649 & 0.936 & 0.957 & \textbf{1.000} & \textbf{1.000} & 0.926\tabularnewline
3\_12\_1\_6 & 240 & 0.672 & 0.933 & 0.908 & \textbf{0.992} & \textbf{0.992} & 0.882\tabularnewline
3\_15\_11\_1 & 390 & 0.649 & 0.948 & 0.974 & 0.984 & \textbf{1.000} & 0.953\tabularnewline
4\_15\_2\_24 & 810 & 0.610 & 0.898 & 0.911 & \textbf{1.000} & \textbf{1.000} & 0.935\tabularnewline
4\_15\_4\_1 & 965 & 0.566 & 0.875 & 0.902 & 0.987 & \textbf{0.996} & 0.919\tabularnewline
5\_15\_2\_12 & 1,350 & 0.565 & 0.862 & 0.896 & 0.990 & \textbf{0.999} & 0.951\tabularnewline
5\_14\_3\_16 & 1,120 & 0.572 & 0.873 & 0.898 & 0.998 & \textbf{1.000} & 0.959\tabularnewline
6\_15\_2\_2 & 5,130 & 0.516 & 0.801 & 0.841 & 0.987 & \textbf{0.998} & 0.961\tabularnewline
6\_15\_2\_1 & 1,890 & 0.546 & 0.847 & 0.877 & 0.995 & \textbf{1.000} & 0.957\tabularnewline
7\_14\_3\_4 & 6,272 & 0.499 & 0.768 & 0.820 & \textbf{1.000} & 0.999 & 0.955\tabularnewline
8\_14\_2\_15 & 9,326 & 0.479 & 0.747 & 0.782 & \textbf{1.000} & \textbf{1.000} & 0.962\tabularnewline
9\_15\_4\_1 & 21,310 & 0.458 & 0.716 & 0.757 & 0.996 & \textbf{1.000} & 0.961\tabularnewline
10\_14\_38\_1 & 15,038 & 0.460 & 0.716 & 0.754 & \textbf{1.000} & \textbf{1.000} & 0.960\tabularnewline
11\_15\_2\_1 & 30,846 & 0.443 & 0.688 & 0.737 & \textbf{1.000} & \textbf{1.000} & 0.961\tabularnewline
\bottomrule
\end{tabular}

\caption{Observed approximation ratios for a few selected manifolds in Lutz's
manifold library. \foreignlanguage{american}{N indicates the number
of simplices in the complex. Cored refers to Coreduction Algorithm,
Red refers to Reduction Algorithm.} For a given input, the best estimated
approximation ratios across all algorithms tested are highlighted
in bold.}

\label{tab:selManifold}
\end{table}

\begin{table}[th]
\centering

\begin{tabular}{cccccc}
\toprule 
\multirow{2}{*}{$D$} & \multirow{2}{*}{No. of complexes of dimension D} & \multirow{2}{*}{Avg size} & \multirow{2}{*}{Worst ratio} & \multirow{2}{*}{Avg ratio} & \multirow{2}{*}{\% optimal}\tabularnewline
 &  &  &  &  & \tabularnewline
\midrule
3 & 166 & 265 & 0.969 & 0.992 & 42.17\tabularnewline
4 & 76 & 630 & 0.979 & 0.995 & 39.47\tabularnewline
5 & 114 & 1,445 & 0.982 & 0.998 & 75.43\tabularnewline
6 & 15 & 3,761 & 0.984 & 0.993 & 26.67\tabularnewline
7 & 33 & 5,988 & 0.996 & 0.999 & 87.88\tabularnewline
8 & 26 & 9,165 & 0.989 & 0.998 & 69.23\tabularnewline
9 & 9 & 14,385 & 0.993 & 0.999 & 66.67\tabularnewline
10 & 2 & 9,566 & 1.000 & 1.000 & 100.00\tabularnewline
11 & 5 & 23,096 & 1.000 & 1.000 & 100.00\tabularnewline
\midrule
All & 446 & 2,271 & 0.969 & 0.995 & 56.05\tabularnewline
\bottomrule
\end{tabular}

\caption{In this table, we summarize the results for all the 446 complexes
in Lutz manifold dataset. The results are grouped row-wise for each
dimension. Avg size refers to the average number of simplices for
complexes of dimension $D$. Worst ratio refers to worst estimated
approximation ratio for Algorithm~3 whereas Avg ratio refers to average
estimated approximation ratio for Algorithm~3, for complexes of dimension
$D$. \% optimal refers to percentage of complexes for which Algorithm~3
computes the optimal Morse matching.}

\label{tab:manifoldSummary}
\end{table}

Table~\ref{tab:manifoldSummary} summarizes the results obtained
using Algorithm~3 for manifolds of different dimensions. We observed
optimal results for~$56\%$ of the complexes. The worst approximation
ratio was observed to be~$0.969$. The descriptions of homology groups
of these complexes are also available in the library. For all the
complexes, we compared the homology computed by application of Morse
matching algorithm followed by boundary operator compuatation and
Smith Normal Form with the ground truth. Our algorithm computes correct
homolgy for all the complexes. The running time of homology computation
was of the order of milli-seconds for most of these complexes.

\subsection{Experiments on random complexes}

We followed the method described by Meshulam and Wallach~\cite{Mesh09}
to generate random complexes. These complexes contain all possible
$d$-simplices for the given number of vertices, for $0\leq d<D$.
However, $D$-simplices are randomly chosen from all possible $D$-simplices
based on probability $p(D)$. We generated two datasets of 100 complexes
each. For each set, we generated a subset of 20 complexes with fixed
$p(D)$, which varies from~$0.1$~to~$0.9$. The number of vertices
was chosen to be $20$~and~$16$ for the $4$~and~$6$ dimensional
datsets, respectively. In Table \ref{tab:randomComplexes}, we report
the results for a single complex selected from each subset. It should
be noted that Algorithm~3 performs well even for random complexes
with non-trivial homology. For Algorithm~3, the worst estimated approximation
ratio over 100 randomly generated 4-dimensional complexes was observed
to be $0.939$. For the 100 randomly generated 6-dimensional complexes
it was observed to be $0.953$. We observed that Algorithm~3 outperformed
reduction and coreduction heuristics for this dataset.

\begin{table}[th]
\centering

\begin{tabular}{ccccccccc}
\toprule 
\multirow{2}{*}{$p(D)$} & \multirow{2}{*}{$N$} & \multirow{2}{*}{$\Sigma\beta_{i}$} & \multicolumn{6}{c}{Estimated approximation ratios}\tabularnewline
\cmidrule{4-9} 
 &  &  & Na\"ive  & Algo 1 & Algo 2 & Algo 3 & Cored & Red\tabularnewline
\cmidrule{1-3} 
\multicolumn{9}{l}{\textbf{Random 6D}}\tabularnewline
\midrule 
0.1 & 16,036 & 3,862 & 0.515 & 0.746 & 0.784 & \textbf{1.000} & \textbf{1.000} & 0.958\tabularnewline
0.3 & 18,324 & 1,574 & 0.477 & 0.739 & 0.765 & \textbf{0.996} & 0.989 & 0.967\tabularnewline
0.5 & 20,612 & 740 & 0.437 & 0.672 & 0.704 & \textbf{0.964} & 0.940 & 0.936\tabularnewline
0.7 & 22,899 & 3,003 & 0.428 & 0.653 & 0.705 & \textbf{0.991} & 0.981 & 0.951\tabularnewline
0.9 & 25,188 & 5,292 & 0.425 & 0.634 & 0.710 & \textbf{1.000} & 0.998 & 0.953\tabularnewline
\midrule
\multicolumn{9}{l}{\textbf{Random 4D}}\tabularnewline
\midrule
0.1 & 7,745 & 2,327 & 0.598 & 0.900 & 0.914 & \textbf{0.999} & 0.996 & 0.979\tabularnewline
0.3 & 10,846 & 800 & 0.470 & 0.704 & 0.727 & \textbf{0.952} & 0.915 & 0.927\tabularnewline
0.5 & 13,947 & 3,877 & 0.465 & 0.692 & 0.732 & \textbf{0.992} & 0.978 & 0.956\tabularnewline
0.7 & 17,047 & 6,977 & 0.452 & 0.669 & 0.739 & \textbf{1.000} & 0.996 & 0.954\tabularnewline
0.9 & 20,148 & 10,078 & 0.463 & 0.671 & 0.752 & \textbf{1.000} & \textbf{1.000} & 0.957\tabularnewline
\bottomrule
\end{tabular}

\caption{This table lists estimated approximation ratios for a selected set
of random complexes of dimensions 6 and 4. Each row represents results
obtained by various algorithms for a single randomly generated instance.
\foreignlanguage{american}{$N$ indicates the number of simplices
in the complex. Cored refers to Coreduction Algorithm, Red refers
to Reduction Algorithm. }$p(D)$ denotes the probability with which
simplices of dimension $D$ are chosen. $\Sigma\beta_{i}$ denotes
the sum of Betti numbers. For a given input, the best estimated approximation
ratios across all algorithms tested are highlighted in bold.}

\label{tab:randomComplexes}
\end{table}

\subsection{Experiments on Type 2 random complexes\label{sub:Type-2-random}}

We also used a variant of the Meshulam and Wallach~\cite{Mesh09}
method for generation of random complexes, where we choose random
number of $d$-simplices for all $d$. The generation of these random
complexes proceed from lowest dimension to highest, and a random simplex
is added to the complex only if all its facets are part of the complex.
We generated a dataset containing 100 5-dimensional complexes with
following parameters: number of vertices was chosen as 40, the probability
of selecting a $d$-simplex is given by the vector $[1,1,0.7,0.9,1,0.9]$.
With these parameters we obtain complexes with non-trivial homology,
as evidenced by their Betti numbers which lie in the range $[1,0,0-1,2945-3658,51-106,0-3]$.
Table~\ref{tab:type2complexes} lists the results for five complexes
selected from this dataset. The worst estimated approximation ratio
over 100 randomly generated 5-dimensional complexes for Algorithm~3
was observed to be 0.989. We again observed that Algorithm~3 consistently
outperformed reduction and coreduction heuristics for all the complexes
in this dataset.

\begin{table}[h]
\centering

\begin{tabular}{cccccccc}
\toprule 
\multirow{2}{*}{$N$} & \multirow{2}{*}{$\Sigma\beta_{i}$} & \multicolumn{6}{c}{Estimated approximation ratios}\tabularnewline
\cmidrule{3-8} 
 &  & Na\"ive  & Algo 1 & Algo 2 & Algo 3 & Cored & Red\tabularnewline
\midrule
\midrule 
39,046 & 3,366  & 0.540 & 0.814 & 0.841 & \textbf{0.991} & 0.987 & 0.979\tabularnewline
39,233  & 3,247  & 0.538 & 0.809 & 0.844 & \textbf{0.993} & 0.982 & 0.977\tabularnewline
39,199  & 3,253  & 0.535 & 0.808 & 0.835 & \textbf{0.992} & 0.984 & 0.978\tabularnewline
39,128  & 3,314  & 0.538 & 0.815 & 0.838 & \textbf{0.994} & 0.986 & 0.979\tabularnewline
39,526  & 3,172  & 0.538 & 0.809 & 0.837 & \textbf{0.991} & 0.985 & 0.979\tabularnewline
\bottomrule
\end{tabular}

\caption{This table lists estimated approximation ratios for a selected set
of 5-dimensional Type 2 random complexes. Each row represents results
obtained by various algorithms for a single randomly generated instance.
\foreignlanguage{american}{N indicates the number of simplices in
the complex. Cored refers to Coreduction Algorithm, Red refers to
Reduction Algorithm.} $\Sigma\beta_{i}$ denotes the sum of Betti
numbers. For a given input, the best estimated approximation ratios
across all algorithms tested are highlighted in bold.}

\label{tab:type2complexes}
\end{table}

\selectlanguage{american}%

\subsection{Discussion on experimental results}

For all datasets we studied, Algorithm~\ref{alg:minfacet} and Coreduction
Algorithm outperform all other algorithms in terms of achieving best
estimated approximation ratios. For the Hachimori dataset and the
Lutz dataset, the coreduction algorithm fares slightly better, whereas
for random datasets, Algorithm~\ref{alg:minfacet} does better. In
general, Algorithm~\ref{alg:minfacet} outperforms all other algorithms
for large sized complexes or when the size of \foreignlanguage{english}{$\Sigma\beta_{i}$
is large. }

\section{Discussion on complexity}

\selectlanguage{english}%
Maximum cardinality bipartite matching is the primary bottleneck for
all the algorithms described in this paper. Graph matching can be
performed in $O(V^{1.5})$ time for Hasse graphs of simplicial complexes
using Hopcroft-Karp algorithm~\cite{HK73}. With appropriate choice
of data structures, all other procedures of all three Algorithms can
be made to run in linear time. 

In particular, for Algorithm ~\ref{alg:minfacet}, we maintain separate
queues for every facet-degree. Consider the graph $\mathcal{G}$ induced
by the min-facet degree simplices. To extract a min-facet component,
we simply find a single connected component within this graph $\mathcal{G}$.
Once the min-facet component is deleted from the $d$-interface, we
update the facet-degrees of all affected simplices within the $d$-interface.
Extraction and maintenance of min-facet components is therefore a
linear time operation. 

Also, for Algorithms~\ref{alg:interface} and~\ref{alg:minfacet},
the approximation ratios do not depend on the graph matching steps.
Graph matching step merely serves the purpose of heuristic improvement.
So, effectively by removing graph matching step from Algorithms~\ref{alg:interface}
and~\ref{alg:minfacet} become linear time approximation algorithms.
But this improvement in computational complexity is at the cost of
estimated approximation ratios observed in practice. 

\selectlanguage{american}%

\section{Conclusions and further work \label{sec:Conclusion}}

We believe that approximation algorithms is the definitive algorithmic
way to study Morse matchings. Our belief is validated by theoretical
results and additionally supported by experimental results where we
get close to optimal ratios even for random complexes. 

In future, we plan to further improve the approximation bounds, remove
dependency on graph matching (for improving estimated approximation
ratios) and develop efficient C++ implementations. In particular,
to obtain dimension independent bounds for Max Morse Matching Problem
remains a challenging open problem.

\section*{Acknowledgements}

The first author is grateful to Michael Lesnick for help and support.
The work is partially supported by the Department of Science and Technology,
India under grant SR/S3/EECE/0086/2012.

\begin{appendices} 
\renewcommand\thetable{\thesection\arabic{table}}   
\renewcommand\thefigure{\thesection\arabic{figure}}

\section{Optimal algorithms for $1$-interface and $D$-interface\label{sec:Optimal-Algorithms-for}}

The optimal algorithm for $D$-interface relies on the Lemma~\ref{lem:dconnected}
which has previously been proved using graph theoretic methods~\cite{JP06,BLPS13}.
See Appendix A of ~\cite{BLPS13}. 

\begin{lem}\label{lem:dconnected} Suppose that following the design of Morse matching for the $(D-1)$-interface, all the $(D-1)$-simplices that are matched are deleted. Then upon execution of this operation the $D$-interface stays connected. 
\end{lem} 

To design the vector field for $1$-interface one may use DFS in
the following way. Pick an arbitrary vertex $s$ as the start vertex
and mark it critical. Then invoke the procedure DFS($s$):

procedure \textbf{DFSoptimal}($v$,$\mathcal{G}$)
\begin{enumerate}
\item Mark $v$ as visited
\item If there exists an edge $\left\langle v,w\right\rangle $ such that
$w$ is not visited then 

\begin{enumerate}
\item match $\left\langle w,\left\langle v,w\right\rangle \right\rangle $ 
\end{enumerate}
\item DFS($w$)
\end{enumerate}
\begin{lem} There exist simple linear time algorithms to compute optimal Morse Matchings for $1$-interface and $D$-interface of $D$-dimensional manifolds.
\end{lem}
\begin{proof}
Since the graph is connected, every vertex will be visited. Also, except for the start node, every other node is matched before it is visited. The edges that are matched belong to the DFS search tree and hence do not form a cycle. Therefore, the only critical vertex is the start vertex. Therefore the simple procedure \textbf{DFSoptimal}() can be used to design optimal gradient vector field for the $1$-interface. Note that the direction of gradient flow for the 1-interface will be exactly opposite of the direction of DFS traversal. \\
We can associate a dual graph to the $D$-interface. Also from Lemma~\ref{lem:dconnected}, we know that following the design of Morse matching for the $(D-1)$-interface and deletion of matched $(D-1)$ simplices, the dual graph remains connected. So, once again, we can use the procedure  on the dual graph. Therefore, upon application of DFS algorithm for the $D$-interface, we will have exactly one critical $D$-simplex which is the start vertex for the DFS and all other $D$-simplices are regular. This algorithm is optimal since input complex $\mathcal{K}$ is a manifold without boundary and hence must have at least one critical $D$-simplex.
\end{proof}

\end{appendices}

\bibliographystyle{plain}

\end{document}